%% file: main.tex
\documentclass[a4paper,UKenglish,cleveref, autoref, thm-restate, subfigure]{lipics-v2021}

\pdfoutput=1 
\hideLIPIcs  

\title{Set Automata and Limits of Decidability of Two-Variable Logic on Data Words} 

\author{Shibashis Guha}{Tata Institute of Fundamental Research}{shibashis@tifr.res.in}{}{Partially supported by CEFIPRA project no.~7302-I and by the Department of Atomic Energy, Government of India, under project no.~RTI4014}

\author{Amaldev Manuel}{Indian Institute of Technology Goa}{amal@iitgoa.ac.in}{}{}

\author{S P {Rishal}}{Sarva Labs}{sprishal@gmail.com}{}{}

\authorrunning{S. Guha, A. Manuel, and S P. Rishal} 

\Copyright{Jane Open Access and Joan R. Public} 



\relatedversion{} 

\nolinenumbers 

\EventEditors{John Q. Open and Joan R. Access}
\EventNoEds{2}
\EventLongTitle{42nd Conference on Very Important Topics (CVIT 2016)}
\EventShortTitle{CVIT 2016}
\EventAcronym{CVIT}
\EventYear{2016}
\EventDate{December 24--27, 2016}
\EventLocation{Little Whinging, United Kingdom}
\EventLogo{}
\SeriesVolume{42}
\ArticleNo{23}

\usepackage{centernot}
\usepackage{cite}
\usepackage{graphicx}
\usepackage{textcomp}
\usepackage{xcolor}
 \usepackage{tikz}
 \usepackage{tikzit}
 \input{styles.tikzstyles}

 \usepackage{stmaryrd}
 \usepackage{mathtools}
 \usepackage{caption}
 \usepackage[disable]{todonotes}
 \usepackage{bbm}
 \usepackage{setspace}
\usepackage[bibliography=common]{apxproof} 
\usepackage[commandnameprefix=always]{changes}
\definechangesauthor[name={Shibashis}, color=teal]{SG}

\def\BibTeX{{\rm B\kern-.05em{\sc i\kern-.025em b}\kern-.08em
    T\kern-.1667em\lower.7ex\hbox{E}\kern-.125emX}}

\input{macro}

\newcommand{\auto}{\mathcal{A}}
\newcommand{\fotwo}{\mathrm{FO}^2}

\newcommand{\bftwo}{\mathbbm{2}}

\newcommand{\sem}[1]{\llbracket#1\rrbracket}

\newcommand{\defeq}{\triangleq}

\makeatletter
\newcommand{\farsim}{\mathbin{\mathpalette\make@circled\sim}}
\newcommand{\make@circled}[2]{%
  \ooalign{$\m@th#1\smallbigcirc{#1}$\cr\hidewidth$\m@th#1#2$\hidewidth\cr}%
}
\newcommand{\smallbigcirc}[1]{%
  \vcenter{\hbox{\scalebox{0.77778}{$\m@th#1\bigcirc$}}}%
}
\makeatother

\newcommand{\EMSO}{\mathrm{EMSO}}

\newcommand{\mach}{\mathcal{M}}

\newcommand{\instr}{\texttt{I}}

\newtheoremrep{lemma}[definition]{Lemma}
\newtheoremrep{proposition}[definition]{Proposition}
\newtheoremrep{fact}[definition]{Fact}

\SetSymbolFont{stmry}{bold}{U}{stmry}{m}{n}

\begin{document}
\maketitle

\begin{abstract}

We extend the two-variable logic on data words\cite{BMSSD11} with guarded regular binary predicates of the form $\tildeL(x,y)$ that is true if positions $x$ and $y$ are in the same class and the factor strictly between $x$ and $y$ is in the regular language $L$. 
We characterise the class of monoids for which the extension of the two-variable logic with guarded predicates recognised by the monoid is decidable, namely the class of idempotent monoids whose two-sided ideals are linearly ordered. 
For this, we introduce an automata formalism, set automata, that is equivalent to the class automata of Boja\'nczyk and Lasota and thus has an undecidable emptiness problem. 
We identify a subclass of set automata called ordered quasi-normal set automata that has a decidable emptiness problem by reduction to the emptiness problem of ordered multicounter automata. 
We show that the two-variable logic extended with guarded regular predicates recognised by a semigroup $S$ is expressively equivalent to a quasi-normal set automaton with the semigroup of transformations $S$. In particular, if $S$ is a linear band monoid then the resulting automaton is ordered, and the decidability result follows. 

\end{abstract}
\input{appendix}
\input{sec1-intro}

\input{sec2-automata}

\input{sec3-osbdecidability}

\input{sec4-logic}
\input{sec5-decidability}

\input{sec6-supplementary}
\input{sec7-conclusion}

\bibliography{references}

\end{document}

%% file: styles.tikzstyles
\tikzstyle{red node}=[fill=red, tikzit category=nodes, shape=circle, draw=black]
\tikzstyle{blue node}=[fill=blue, shape=circle, draw=black, tikzit category=nodes]
\tikzstyle{black node}=[fill=black, shape=|, draw=black, tikzit category=nodes]
\tikzstyle{green node}=[tikzit fill=green, fill=green, shape=circle, draw=black, tikzit category=nodes]
\tikzstyle{yellow square}=[draw=black, fill=yellow, shape=rectangle]
\tikzstyle{blue node 2}=[fill={rgb,255: red,128; green,0; blue,128}, draw=black, shape=circle, tikzit fill=blue]

\tikzstyle{edge}=[-, dashed]
\tikzstyle{blue pointer}=[->, draw=blue]
\tikzstyle{black pointer}=[->, draw=black]
\tikzstyle{blue dashed}=[-, draw=blue]
\tikzstyle{arrow}=[<->, draw=black]
\tikzstyle{length}=[|-|, draw=black]
\tikzstyle{leftlength}=[|-, draw=black]

%% file: macro.tex
\newcommand{\tilL}{\widetilde{L}}

\newcommand{\id}{\mathrm{id}}
\newcommand{\str}{\mathit{str}}

\newcommand{\barP}{\overline{P}}
\newcommand{\barQ}{\overline{Q}}

\newcommand{\barU}{\overline{U}}
\newcommand{\barV}{\overline{V}}

\newcommand{\barX}{\overline{X}}
\newcommand{\barY}{\overline{Y}}
\newcommand{\barZ}{\overline{Z}}

\newcommand{\bars}{\overline{s}}

\newcommand{\baru}{\overline{u}}
\newcommand{\barv}{\overline{v}}

\newcommand{\barz}{\overline{z}}

\newcommand{\bfA}{\mathbf{A}}

\newcommand{\bfD}{\mathbf{D}}

\newcommand{\bfR}{\mathbf{R}}

\newcommand{\bfT}{\mathbf{T}}

\newcommand{\bfX}{\mathbf{X}}

\newcommand{\bfv}{\mathbf{v}}

\newcommand{\calA}{\mathcal{A}}
\newcommand{\calB}{\mathcal{B}}
\newcommand{\calC}{\mathcal{C}}
\newcommand{\calD}{\mathcal{D}}

\newcommand{\calF}{\mathcal{F}}

\newcommand{\calH}{\mathcal{H}}

\newcommand{\calJ}{\mathcal{J}}
\newcommand{\calK}{\mathcal{K}}
\newcommand{\calL}{\mathcal{L}}
\newcommand{\calM}{\mathcal{M}}
\newcommand{\calN}{\mathcal{N}}

\newcommand{\calP}{\mathcal{P}}

\newcommand{\calR}{\mathcal{R}}
\newcommand{\calS}{\mathcal{S}}

\newcommand{\calU}{\mathcal{U}}

\newcommand{\calX}{\mathcal{X}}
\newcommand{\calY}{\mathcal{Y}}
\newcommand{\calZ}{\mathcal{Z}}

\newcommand{\bfone}{\mathbf{1}}

\newcommand{\FO}{\mathrm{FO}}

\newcommand{\emsotwo}{\mathrm{EMSO}^2}

\newcommand{\tildeL}{\widetilde{L}}

\newcommand\sg[1]{}
\newcommand\sgchanged[1]{}
\newcommand\am[1]{}
\newcommand\amchanged[1]{}

\newif\ifdraft
\drafttrue
\ifdraft
\renewcommand\sg[1]{\todo[inline,size=\scriptsize,backgroundcolor=Yellow]{#1 - \textbf{Stefan}}}

\renewcommand\sgchanged[1]{{\color{red}{#1}}}
\renewcommand\am[1]{\todo[inline,size=\scriptsize,backgroundcolor=Yellow]{#1 - \textbf{Amal}}}
\renewcommand\amchanged[1]{{\color{blue}{#1}}}
\else
\fi

%% file: appendix.tex

\begin{toappendix}
\section{Appendix}

\paragraph*{Orders Preliminaries}

A \emph{preorder} $\lesssim$ on a set $X$ is a binary relation on $X$ that is reflexive and transitive. 
A preorder $\lesssim$ is \emph{total} if $x\lesssim y$ or $y\lesssim x$ for each $x,y\in X$.
We say the subset $Y \subseteq X$ is \emph{$\lesssim$-total} if any two elements in $Y$ is comparable with respect to $\lesssim$, in other words $Y$ is totally preordered by $\lesssim$.
A \emph{partial order} is an antisymmetric preorder.
A \emph{linear order} is a total partial order. 

\paragraph*{Semigroup Preliminaries}\label{app:semigroup} 

We recall some fundamental definitions from the local theory of semigroups. 
A detailed account can be found in standard textbooks, such as \cite{pinbook, lallementbook}.

For a semigroup $S$, its monoidal extension $S^1$ is defined to be $S$ if it is already a monoid and otherwise to be the monoid obtained by adjoining an identity element $1$ to it.
We denote by $\langle X\rangle$ the semigroup generated by elements in $X$. 
For semigroups $S$ and $T$, let $S\leq T$ denote that $S$ is a subsemigroup of $T$. 
An element $e$ of a semigroup is said to be an \emph{idempotent} if $e^2 = e$.  
In a finite semigroup, for every element $x$ there exists a positive integer $k$ such that the element $x^k$ is an idempotent. 
Moreover there is some $n>0$ such that the $s^n$ is an idempotent for each element $s$ of the semigroup. 
The least such $n$ is called the exponent of the semigroup and is denoted by $\omega$. 
A semigroup is \emph{aperiodic} if $s^{\omega+1}=s^\omega$ for each element $s$ of the semigroup.
Let $S$ and $T$ be semigroups. The semigroup $T$ is a \emph{quotient} of $S$ if there exists a surjective morphism from $S$ onto $T$. The semigroup $T$ is said to \emph{divide} $S$ if $T$ is a quotient of a subsemigroup of $S$.

\begin{definition}[Green's relations]
Let $s,t$ be elements of a semigroup $S$. 
The Green's preorders $\leq_\calR, \leq_\calL$, and $\leq_\calJ$ are given as follows: $s \leq_\calR t ~\defeq~s=ty$ for some $y\in S^1$, $s \leq_\calL t  ~\defeq~s=xt$ for some $x\in S^1$, and $s \leq_\calJ t  ~\defeq~s=xty$ for some $x,y\in S^1$.

Let $ \calK\in \{\calR, \calL, \calJ\}$. 
The Green's preorders have equivalence relations associated with them that are given as $s \calK t~\defeq~s\leq_\calK t \text{ and } t\leq_\calK s$
Further, the relation $\calH$ is given by $s \calH t ~\defeq~s\calL t \text{ and } s\calR t$. 
\end{definition}  

The equivalence class of an element $s$ under the relation $\calK\in \{\calR, \calL, \calJ, \calH\}$ is called the $\calK$-class of $s$ and is denoted by $\calK(s)$. 
Let $\calK'\in\{\calR, \calL, \calJ\}$.
Two $\calK'$-classes are \emph{comparable} when every element in one class is comparable to every element in the other class in the $\leq_{\calK'}$ preorder.
In finite semigroups the relations $\calD = \calR \circ  \calL = \calL \circ  \calR$ and $\calJ$ coincide, and we use them interchangeably.

\begin{definition} [Ideals] \label{def:ideals}

A (two-sided) \emph{ideal} of a semigroup $S$ is a subset $I$ of $S$ such that $S^1IS^1\subseteq I$, that is, for each $s\in I$ and $x,y\in S^1$ we have that $xsy\in I$. 
\end{definition}
It is trivial to see that for an ideal $I$ of a semigroup $S$, we have $S^1IS^1=I$ since $1I1 = I \subseteq S^1$.

\begin{fact}[Prefix/Suffix Lemma]\label{Fact:fallToRLClass} Let $x,y$ be elements of a \emph{finite} semigroup. If $x\leq_\calL y$ and $x \calJ y$ then $x\calL y$. If $x\leq_\calR y$ and $x\calJ y$ then $x\calR y$. 
\end{fact}

\begin{fact}[Stability of $\leq_\calL$ and $\leq_\calR$] \label{fact:stable}The following is true in any semigroup.
\begin{enumerate}
\item $\leq_\calR$ is stable on the left, i.e., if $x\leq_\calR y$ then $zx\leq_\calR zy$. Hence, if $x \calR y$ then $zx\calR zy$. 
\item $\leq_\calL$ is stable on the right, i.e., if $x\leq_\calL y$ then $xz\leq_\calL yz$. Hence, if $x \calL y$ then $xz\calL yz$.
\end{enumerate}
\end{fact}

\begin{fact}\label{Fact:uniqueIdemH}
If $x$ and $y$ are idempotents and $x\calH y$, then $x=y$. 
\end{fact}


\begin{fact}[Location Theorem]\label{Fact:location}
Let $s$ and $t$ be elements of a finite semigroup $S$ such that $s\calD t$. 
The following are equivalent:
\begin{enumerate}
\item $\mathit{st}\ \calD\ s$,
\item $\mathit{st}\in \calR(s)\cap \calL(t)$,
\item $\calL(s)\cap \calR(t)$ contains an idempotent.
\end{enumerate}
\end{fact}

\end{toappendix}

%% file: sec1-intro.tex
\section{Introduction}
A \emph{data word} 
is a finite sequence of pairs $(a_i,d_i)\in \Sigma\times \calD$ where $\Sigma$ is a finite alphabet and $\calD$ is an infinite domain of \emph{data values} that can be tested only for equality. A set of data words $L$ is called a \emph{data language}. 
In most cases (including this work) $L$ is assumed to be closed under permutations of the set $\calD$, that is, every occurrence of some data value $d_i$ can be replaced with another data value $d_j \neq d_i$ as long as such a mapping is a permutation.
Thus the specific data values considered are not important to the study and we only capture properties involving equality of data values.
For our purposes, $L$ can be represented as a collection of relational first-order structures of the form $\left([n], (a)_{a\in \Sigma}, <, +1 \sim\right), \moo1$ where $[n]$ and $<$ denote the set $\{1,\ldots,n\}$ and the natural order on it, the unary predicates $(a)_{a\in \Sigma}$ denotes the labelling by $\Sigma$, and $\sim$ is the equivalence relation on positions given by data values, i.e., $i \sim j$ if $d_i = d_j$.  
The relation $\moo1$ denotes the \emph{class successor} relation, i.e., $i \moo1 = j$, if $i\sim j$ and the positions strictly between $i$ and $j$ are not equivalent to them. 
The \emph{class} of $i$ is the equivalence class of $i$ under the relation $\sim$, i.e., the set of all positions labelled with the data value $d_i$. 


The primary direction in the study of data words has been identifying suitable notions of automata and logics for them (See \cite{NSV04,BS07} and the survey \cite{Segoufin06}). The models and logics in the literature can broadly be classified into two families: temporal logics (correspondingly, the register automata family) and first-order logics (data automata family). The prominent member of the first family is the freeze-LTL (i.e., LTL with registers) of \cite{DL09}, while the most important logic of the latter is the two-variable first-order logic \cite{BMSSD11}. In this work we present a decidable, strict extension of the two-variable logic that incorporates some aspects of temporal logics. First we recall the two-variable logic in some detail.

For a vocabulary $\tau$ of unary and binary predicates let $\fotwo[\tau]$ denote the first-order formulas using the predicates from $\tau$ and the variables $x$ and $y$. 
The logic $\emsotwo[\tau]$ consists of formulas of the form $\exists X_1 \ldots \exists X_n \, \varphi$ with $\varphi \in \fotwo[\barX, \tau]$, where $\barX$ stands for the monadic second-order predicates $X_1,\ldots, X_n$ that quantify over sets of positions. 
Certainly, it is natural to consider the first-order logic ($\FO$) with the vocabulary $(\Sigma, <,\sim)$ on data words.
However, as one would expect, the satisfiability problem of the logic $\FO[\Sigma, <, \sim]$ is undecidable \cite{BMSSD11}. 
The undecidability persists even if the number of variables is restricted to three. 
This prompted the study of the two-variable fragment of the logic. 
For convenience in working with automata, we consider the extension of the first-oder logics with existential monadic second-order predicates. 
It is shown in \cite{BMSSD11} that the satisfiability problem of $\emsotwo[\Sigma, <,+1, \sim, \moo1]$ is decidable over data words in non-elementary time. 
The proof is by translating the formulas into an automaton formalism called data automata and showing the decidability of the emptiness problem of data automata. 
As far as the study of logics over data words is concerned this result is a milestone and despite nearly two decades of active research, $\FO^2$ remains as one of the yardstick logics with which any new formalism is measured.

Next we look at some examples. The \emph{string projection} of a data word $w$ is the word $\mathrm{str}(w) = a_1\cdots a_n$. Let $1 \leq i_1<i_2 < \cdots < i_k\leq n$ be a class of $w$. The string projection of this class is the word $u=a_{i_1} \cdots a_{i_k}$. We say $u$ is a class of $w$ to mean 
a class whose string projection is $u$.     

\begin{example}\label{ex:IDZ}
Let $\Sigma=\{\iota, \delta, z\}$. We can think of $\iota$ and $\delta$ as increments and decrements of a counter and $z$ as a zero-test.
\begin{enumerate}
\item Let $L_1$ be the set of all data words where each class is of the form $\iota\delta$ or $z$. This language is defined by the formula $\varphi_1$ where
\begin{align}
\varphi_1 &:= \forall x \forall y\, ( x <y \wedge x \sim y  \rightarrow \iota(x) \wedge \delta(y)) \wedge ( (x \sim y  \rightarrow x = y) \rightarrow z(x))
\end{align}
\item Let $L_2$ consists of all data words satisfying the property: there is no position labelled by $z$ between each $\iota$ and $\delta$ of the same class. 
\end{enumerate}
\end{example}

The language $L_2$ is not definable in $\FO^2$ as noted in \cite{Segoufin06}.
In fact it is not recognised by the data automata of \cite{BMSSD11}, that is equivalent to the logic $\EMSO^2[\Sigma, <,x +1, \sim, \moo1]$ 
However, $L_2$ is definable in freeze-LTL with only forward-looking modalities (Next and Until), a logic with a decidable satisfiability problem (as a trade-off this logic cannot define $L_1$, see \cite{Segoufin06}).    

To define $L_2$ in $\FO^2$, we extend the two-variable logic with the below predicates.

\begin{definition}[Regular Predicate]
A \emph{regular predicate} over the alphabet $\Sigma$ is a binary relation of the form $L(x,y)$ where $L\subseteq \Sigma^*$ is a regular language.
The word $w=a_1\cdots a_n\in \Sigma^*$ at the pair of positions $(i,j)$, for $i,j \in[n]$, satisfies the predicate $L(x,y)$ if $i< j$ and the factor $a_{i+1}\cdots a_{j-1}$ given by the interval $(i,j)$ is in $L$.
A \emph{guarded} regular predicate  
$\widetilde{L}(x,y)$ is equivalent to the formula $x\sim y \wedge L(x,y)$.  
\end{definition}

For a regular expression $r$, let $\sem{r}$ denote the language defined by it. 
Using regular predicates, the language $L_2$ can be defined in $\FO^2$ by the below formula.
\begin{equation}
\label{eq:onecounter-reg}
\forall x \forall y\, \left( ( \iota(x) \wedge \delta(y) \vee \delta(x) \wedge \iota(y))  \wedge x \sim y \rightarrow \neg \widetilde{\sem{\Sigma^*z\Sigma^*}}(x,y)\right).
\end{equation}

The below example illustrates that adding regular predicates to the vocabulary leads to undecidability even when relatively simple regular predicates are used. 




\begin{example}\label{ex:2IDZ}
Let $\Sigma'=\{\iota_1, \delta_1, z_1, \iota_2, \delta_2, z_2\}$. We can think of the alphabet as encoding the increments, decrements, and zero-tests of a two-counter machine. 
\begin{enumerate}
\item Let $L_3$ be the set of all data words satisfying the following property: each class is of the form $\iota_i\delta_i$ or $z_i$, $i\in \{1,2\}$. Clearly $L_3$ is in $\FO^2[\Sigma,<,+1,\sim]$.
\item Let $L_4$ consist of all data words satisfying the following property: for each $i\in \{1,2\}$, there is no position labelled by $z_i$ between each $\iota_i$ and $\delta_i$ of the same class. 
\end{enumerate}
\end{example}

From Formula \ref{eq:onecounter-reg}, it is not difficult to see that $L_4$ is definable in $\FO^2$ using the guarded regular predicates defined by the languages $\Sigma'^*z_1\Sigma'^*$ and $\Sigma'^*z_2\Sigma'^*$. This is sufficient to encode the run of a two-counter machine as a data word.

The algebraic theory of regular languages where the semigroup recognising a language is considered, allows for properties of regular languages to be understood through the structure of the semigroup recognising it. 
Motivated by this, in the logic, we consider families of regular predicates defined by a semigroup recognising it.

A \emph{semigroup} $S$ is a set with an associative binary operation. It is a \emph{monoid} if the operation has an identity, denoted by $1$. All semigroups we consider in this paper are finite or free semigroups of the form $\Sigma^*$. 
Let $M$ be a finite monoid. A morphism $h:\Sigma^* \rightarrow M$ is a map satisfying $h(ab)=h(a)h(b)$ and $h(\varepsilon)=1$. 
A language $L\subseteq \Sigma^*$ is \emph{recognised} by a morphism $g:\Sigma^* \to M$, where $M$ is a finite monoid, if there is a finite subset $P \subseteq M$ such that $L = h^{-1}(P)$. 
We also consider semigroups recognising languages. Let $S$ be a finite semigroup and $S^1$ be the monoid obtained by adjoining an identity $1$ to $S$. The monoid morphism $h:\Sigma^*\to S^1$ is \emph{unit-reflecting} if $h^{-1}(1)=\varepsilon$. We say a language $L\subseteq \Sigma^*$ is recognised by $S$ if it is recognised by a unit-reflecting morphism into $S^1$.
A monoid $M$ \emph{recognises} a family $\calF$ of languages over $\Sigma$ if there is a morphism $h: \Sigma^* \to M$ that recognises each language in the family. Definitions of other semigroup-theoretic notions required for our purposes can be found in \Cref{app:semigroup}.

Let $h:\Sigma^* \to M$ be a morphism. The logic $\FO^2[\Sigma, <, +1, \sim, \moo1, \calF_h]$ is the extension of the logic with the family of regular predicates $\calF=\{L_1, \ldots, L_k\}$ such that $h$ recognises the family $\calF$. For a monoid $M$, by $\FO^2[\Sigma, <, +1, \sim, \moo1, \calF_M]$ we denote the set of formulas $\varphi \in \FO^2[\Sigma, <, +1, \sim, \moo1, \calF_h]$ for some morphism $h:\Sigma^*\to M$.
The logics $\FO^2[\Sigma, <, +1, \sim, \moo1, \widetilde{\calF}_h]$ and 
$\FO^2[\Sigma, <, +1, \sim, \moo1, \widetilde{\calF}_M]$ are defined analogously where 
the predicates $\widetilde{\calF}_h=\{\tildeL_1, \ldots, \tildeL_k\}$ used are guarded. The above definitions extend naturally to semigroups and unit-reflecting morphisms.

\begin{example}\label{ex:monoidrecog}
    The language family $\{\Sigma^* z\Sigma^*\}$ from \Cref{eq:onecounter-reg} is recognised by the two element monoid with a zero $U_1=\{1,0\}$ under the usual product operation, with the morphism $z\mapsto0$ and $\iota,\delta\mapsto1$ and accepting set $P = \{0\}$.  
    However, the family $\calF=\{\Sigma'^*z_1\Sigma'^*, \Sigma'^*z_2\Sigma'^*\}$ is not recognisable by the monoid $U_1$ as both languages cannot be recognised by the same morphism. 
    However  $\calF$ is recognised by the product monoid $U_1^2= U_1\times U_1 = \{(1,1), (1,0), (0,1), (0,0)\}$.
    under the morphism $z_1\mapsto (0,1)$, $z_2\mapsto (1,0)$ with the accepting sets $\{(0,1), (0,0)\}$ and $\{(1,0), (0,0)\}$ respectively. 
\end{example}

A key difference in the structure of the two monoids above is that in $U_1$, the two-sided ideals (See \Cref{def:ideals}) are linearly ordered while it is not the case in $U_1^2$ (See \cref{fig:u1-u1sqr-n2}). 
Motivated by this, we consider the class of idempotent monoids whose two-sided ideals are linearly ordered called \emph{linear bands}, that is, for all two-sided ideals $I, J$, we have that $I\subseteq J$ or $J\subseteq I$. The idempotency restriction follows from another such language that leads to undecidability (See \Cref{prop:undec}).  


The below example shows how guarded regular predicates parametrised by a family of monoids can be used to capture a previously known decidable logic introduced in \cite{BMSSD11}. 

\begin{example}[Nilpotent Predicates]
Consider the monoid $\{0,1,\ldots, n-1, \infty\}, n\geq 1$ with the below operation where $x+y$ is the usual addition over natural numbers.
\[ x\cdot y = \begin{cases} x+y & \text{if $x,y \in \{0,1,\ldots, n-1\}$ and $x+y < n$}, \\
\infty & \text{otherwise.}
\end{cases}
\]

The unique morphism given by the map $g:\Sigma \mapsto 1$ (and $g:\Sigma \mapsto \infty$ if $n=1$)  recognises the family of languages $\{\Sigma^i \mid 1\leq i \leq n-1\} \cup \{\Sigma^{\geq n}\}$. Thus the logic $\FO^2[\Sigma,<,\sim, \calF_g] \equiv \FO^2[\Sigma,<, +1, +2, +3, \ldots, +n,  \sim]$. The decidability of the latter logic is shown in \cite{BMSSD11}. 
\end{example}

\paragraph*{Our Contributions}




Our main result is as follows.

\begin{restatable}{thm}{repeatedthm}
\label{thm:emso-linband-dec}
Satisfiability of $\emsotwo[\Sigma, <, +1, \sim, \moo1, \widetilde{\calF}_M]$ formulas over data words is decidable if and only if $M$ is a linear band. 
\end{restatable}

Our decidable fragment is a strict extension of the $\fotwo$ of \cite{BMSSD11} as the language $L_2$ of \cref{ex:IDZ} is definable in the decidable fragment of \Cref{thm:emso-linband-dec} but not in the $\fotwo$ of \cite{BMSSD11}. 
The guardedness of the regular predicates does not impose a restriction on the expressibility in the following sense: for every formula in 
$\emsotwo[\Sigma, <, +1, \sim, \moo1, {\calF}_M]$ with regular predicates recognised by a monoid $M$ there is an equivalent formula in $\emsotwo[\Sigma, <, +1, \sim, \moo1, \widetilde{\calF}_{M'}]$ with guarded regular predicates recognised by another monoid $M'$.  
We state our results for the guarded regular predicates. 


To show decidability of the logic, we employ a new automaton formalism called set automata.
A \emph{set automaton} $\calA$ is a finite state automaton with a fixed number of sets $\calX =\{X_1,\ldots, X_k\}$ that can store data values. During the run, the automaton can add or remove the current data value from the sets as well as assign a union of sets to each set (i.e. of the form $X_i = \bigcup_{j \in J} X_j$ for $J \subseteq [k]$). The latter kind of updates can be represented as an element of a semigroup $S$ of relations on the sets. The acceptance of a run is determined by the state as well as the sets in which the data values are present at the end of the run. 
The set automaton model is equivalent to the class automata of \cite{BL10} and is thus undecidable. Also, it is noteworthy that data automata are equivalent to set automata whose semigroup of relations is trivial, that is, the only set update is used is the identity relation.

We identify a subclass of set automata with decidable emptiness, called ordered-quasi normal set automata.  
The decidability of emptiness of ordered quasi-normal set automata is shown by a reduction to the emptiness of ordered multicounter automata\cite{Rein08}, which is known to be decidable.
The complexity of the decision procedures for the logic $\emsotwo[\Sigma, <, +1, \sim, \moo1, \widetilde{\calF}_M]$ and set automata are non-elementary that is inherited from $\fotwo[\Sigma,<,+1,\sim]$ and data automata.

We obtain decidability of the logic by converting $\emsotwo[\Sigma, <, +1, \sim, \moo1, \widetilde{\calF}_M]$ formulas where $M$ is a linear band to ordered quasi-normal set automata. 
To show undecidability of the logic when $M$ is not a linear band, we use a reduction from the halting problem of two-counter machines which is known to be undecidable \cite{Minsky67}.


\paragraph*{Related Works}

\subparagraph*{Automata on Data Words.}

{There exists numerous extensions of finite state automata with registers \cite{KF94}, pebbles \cite{NSV04},  hash tables \cite{BS07}, counters \cite{MR11}, sets \cite{BCG23, GL22}, etc.~to handle data values (See the survey \cite{Segoufin06}). 
The bibliography is extensive and we mention only those relevant to our work.}  
Register automata\cite{KF94} and its variants \cite{NSV04,DL09,BS20} equip finite state automata with a fixed number of registers for storing data values and comparing them using equality.  
Set augmented finite automata \cite{BCG23}, register set automata \cite{GL22}, and history register automata  \cite{GT16} are generalisations of register automata with unbounded registers (or sets in the present terminology). 
The latter two are equivalent to set automata except for the acceptance conditions. \todo{say that they are undec, and that our decidability result is new. check papers}
Class automata \cite{BL10} are closely related to data automata and are expressively equivalent to set automata. A subclass of class automata captures the alternating one register automata of \cite{DL09}. 
The decidable subclass presented in our work was previously unknown.

\subparagraph*{Logics on Data Words.}
As for $\fotwo$ on data words, the results in \cite{BMSSD11} are already mentioned. The related, but weaker logic $\mathrm{EMSO}^2(\Sigma, \sim, +1)$ was studied in \cite{KST12}. It was shown to be equivalent to a semantic restriction of data automata, called weak data automata. 
A weakening of MSO logic with data equality tests called rigidly guarded MSO was introduced in \cite{CLP11}. It was shown that it is as expressive as orbit finite data monoids \cite{Boj11} and its satisfiability is decidable. A $\mu$-calculus that has modalities corresponding to successor, predecessor, class successor and class predecessor relations was introduced in \cite{CM15}.


\paragraph*{Organisation of the Paper}
The paper is organised as follows. \Cref{sec:automata} introduces set automata and its various restrictions namely normal, quasi-normal, and ordered automata. We also prove the necessary closure properties required to translate guarded $\emsotwo$ formulas to quasi-normal set automata. 
Subsequently, in \Cref{sec:logic} we establish the automata-logic connection. 
\Cref{sec:decidability} details the decidability and undecidability results concerning the guarded fragment. 
\Cref{sec:discussion} discusses alternative accepting conditions for set automata, translations between class automata and set automata. 
In \Cref{sec:conclusion} we conclude. 
Standard facts from the theory of Green's relations on finite semigroups, and definitions of basic relations like preorder, partial order, and other related notions are given in Appendix~\ref{app:semigroup}. 
Finally, to improve readability, we have moved some proofs to the Appendix.

%% file: sec2-automata.tex
\section{Set Automata}\label{sec:automata}
\nosectionappendix

\begin{toappendix}
   \subsection*{Proofs for \Cref{sec:automata}}
\end{toappendix}

In this section, we introduce set automata, and the subclasses normal, quasi-normal, and ordered set automata. 
As our main result, we show that the emptiness problem of ordered normal set automata is decidable. 

\subsection{Relations and Transformations}
First, we recall some preliminaries used in the rest of the paper. 

For $k>0$, let $[k]$ denote the set $\{1,\ldots,k\}$.  The set of Boolean values is denoted by $\bftwo=\{0,1\}$. Let $Y$ be a finite ordered set. A column vector $\baru=(u_y)_{y\in Y}$ indexed by $Y$ is an element of $\bftwo^Y$.  
For a subset $Y' \subseteq Y$, the restriction of $\baru$ to $Y'$ is given by $\baru\restriction_{Y'} = (u_y)_{y\in Y'}$.

The set of column vectors over $\bftwo$ indexed by $Y$ is denoted by $\bftwo^Y$. The Boolean operations can be extended to elements of $\bftwo^Y$ pointwise. 
For $u \in \bftwo^Y$, we denote by $u^c = (1-u_y)_{y \in Y}$ the complement of $u$. For $y\in Y$, let $e_y\in\bftwo^Y$ denote the vector whose $y$-component is $1$ and all other components are $0$. The vectors $\{e_y \mid y \in Y\}$ are called unit vectors.

Let $S$ be a set. For a subset $X\subseteq S$, let $\bfone_X:S \to \bftwo$ denote the characteristic function given by $\bfone_X(x) =1$ if $x \in X$ and  $\bfone_X(x) =0$ otherwise. 
Extending this notation, for $\barX = (X_y)_{y\in Y} \in \calP(S)^Y$ a vector of subsets of $S$ indexed by $Y$, let $\bfone_{\barX}:S \to \bftwo^Y$ denote the function $x \mapsto (\bfone_{X_y}(x))_{y \in Y}$. 
The vector $\bfone_{\barX}(x)$ is called the characteristic vector of $x$ in $\barX$. 


A \emph{binary relation} $\rho$ on a set $Y$ is a subset of the cartesian product $Y\times Y$. 
We simply write relation instead of binary relation.
The \emph{converse} of a relation $\rho$ is the relation $\rho^{-1} = \{(y,x)\mid (x,y)\in \rho\}$ where the order is switched. For a subset $Y' \subseteq Y$, the restriction of $\rho$ to $Y'$ is $\rho \restriction_{Y'}=\rho \cap (Y'\times Y')$.
The \emph{image} of an element $x \in Y$ under the relation $\rho$ is the set $\rho(x)\subseteq Y$ given by $\rho(x) = \{ y \in Y \mid (x,y) \in \rho\}$.
Similarly, the image of a subset $S$ of $Y$ under the relation is given by $\rho(S) = \bigcup_{x\in S} \rho(x)$.
A \emph{transformation} on a set is a function from the set to itself. 
A bijective transformation is called a \emph{permutation}. Given relations $\rho$ and $\sigma$ on the set $Y$, their \emph{relation composition \emph{or} product} $\rho\sigma$ is given by
$
    \rho\sigma = \{ (x,z) \in Y \times Y \mid \exists y \in Y, (x,y)\in \rho, (y,z) \in \sigma\}$.
\begin{example}
Let $\rho =\{ (0,0),(0,1),(1,1)\}$ and $\sigma = \{(0,0),(1,0)\}$ be relations on the set $Y=\{0,1\}$.
Then $\rho\sigma = \{(0,0),(1,0)\}$ and $\sigma\rho = \{(0,0),(0,1),(1,0),(1,1)\}$.
\end{example}
The omission of $\circ$ to denote the product operation is intentional. When $\rho$ and $\sigma$ are functions then their \emph{function composition} $\rho \circ \sigma$ is the function $\rho \circ \sigma = \sigma\rho$ that is different from their product $\rho\sigma$.
The product of relations is an associative operation and the relations on a set $Y$ form a monoid with the identity function $\id_Y=\{(y,y) \mid y \in Y\}$ as the identity element. 
Let $\bfR_Y$ and $\bfT_Y$ denote the monoid of relations and transformations on the set $Y$, respectively. 

A relation $\rho$ over $Y$ is naturally viewed as a Boolean matrix $M(\rho)=(m_{ij})_{(y_i,y_j)\in Y \times Y}$ where $m_{ij}=1$ if $(y_i,y_j) \in \rho$ and $0$ otherwise. 
Conversely, for every Boolean matrix $M$ there is a corresponding relation $\rho(M)$ such that $M(\rho(M))=M$. It is easy to verify that $M(\rho\sigma)=M(\rho) \cdot M(\sigma)$.
We denote by $M^T$ the transpose of a matrix $M$. 


\subsection{Set Automaton}



\begin{definition}[Set automaton]
A \emph{set automaton} $\calA$ is a tuple $(Q, Y, \Sigma, \Delta, I, F, C)$ where $Q$ is the finite set of states, $Y$ is a finite set of names for sets, and $\Sigma$ is the input alphabet. 
The transition relation is given by $\Delta \subseteq Q \times \Sigma \times \bftwo^Y \times \bfR_Y \times \bftwo^Y \times \bftwo^Y \times Q$. The initial and final sets of states are respectively $I \subseteq Q$ and $F \subseteq Q$.  
Finally, $C \subseteq \bftwo^Y$ is the family of accepting vectors of membership of data values.
\end{definition}

A \emph{configuration} of the automaton $\calA$ is a pair $\gamma=(q, \barX)$ where $q\in Q$ is a state and $\barX = (X_y)_{y\in Y} \in \calP(\calD)^Y$ is a vector of subsets of data values. A data value $d$ is said to be \emph{present} in $\barX$ if it is in $X_y$ for some $y\in Y$. A configuration is \emph{initial} if $q \in I$ and  $X_y = \emptyset$ for each $y \in Y$. 

When the automaton is in a configuration $\gamma$, on the input pair $(a,d) \in \Sigma\times\calD$, the transition $ (p, \ell, \barz, \rho, \baru, \barv, p') \in \Delta$
is applicable if $p=q$, $\ell=a$, and $\barz= \bfone_{\barX}(d)$, that is, the characteristic vector of membership of the data value $d$ is $\barz$. The \emph{set} or \emph{global update} is given by the relation $\rho\in \bfR_Y$. 
The global update $\rho$ removes the contents of each set $x$ and copies it to each set $y\in \rho(x)$.   
This results in the vector of subsets of data values stored in the sets of $\calA$, $\overline{X'} =  (X_y')_{y\in Y}$ where $X_y'=\bigcup_{x\in \rho^{-1}(y)}X_x$. 
Further, for each data value $d'\in\calD$, 
we obtain the characteristic vector of membership $\bfone_{\overline{X'}}(d')=M(\rho)^T \cdot \bfone_{\barX}(d')$ after the global update. 
The \emph{local updates} are given by the vectors $\baru,\barv \in \bftwo^Y$. 
The local updates are applied on the contents of the sets given by $\overline{X'}$, with the current data value $d$ added to the sets given by $\baru$ and removed from the sets given by $\barv$. 
Let $\barU, \barV  \in \calP(\{d\})^Y$ be such that $\baru=\bfone_{\barU}(d)$ and $\barv  =\bfone_{\barV}(d)$.
This results in the configuration $(p',\bfX'')$ where $\overline{X''} =(X_y'')_{y\in Y} $ is given by $X_y'' =  (X_y' \cup U_y ) \setminus V_y$ for each $y \in Y$.

A configuration $(q,\barX)$ is \emph{accepting} if $q\in F$ and the characteristic vector of each data value present in $\barX$ is in $C$.
A \emph{successful} run of $\calA$ on a data word $w$ is a sequence of applicable transitions taking the automaton from an initial to an accepting configuration. The \emph{language} of $\calA$, denoted as $L(\calA)$, is the set of all data words $w$ on which $\calA$ has a successful run.

Let $\calU(\calA)$ denote the set of global updates used in the transitions of the set automaton $\calA$. 
Let  $M_\calA=\langle \calU(\calA)\rangle$, called the \emph{update monoid of $\calA$}, be the submonoid of relations on $Y$ generated by $\calU(\calA)$. Clearly $M_\calA$ is a submonoid of $\bfR_Y$, i.e.,  $M_\calA \leq \bfR_Y$.

\begin{example}\label{ex:12IDZ} 

Consider the language $L_{12} = L_1 \cap L_2$ over the alphabet $\Sigma=\{\iota, \delta, z\}$ where $L_1$ and $L_2$ are the languages described in \cref{ex:IDZ}.
This language is accepted by a set automaton $\calA$ with three sets $Y_1, Y_2$ and $Y_3$.
At any given position, the set $Y_1$ contains all the data values $d$ such that $(\iota,d)$ was read but a corresponding $(\delta, d)$ has not yet been observed. The set $Y_2$ is used to track whether a $z$ was observed between an $\iota$ and $\delta$ of the same class. In an accepting run, the set $Y_2$ always remains empty. 
The set $Y_3$ is used as a sink to track the data values of each class satisfying either of the above properties that have occurred in the run. The set automaton rejects whenever the current data value is present in $Y_3$ as this signifies that a class is not of the form $\iota$ followed by $\delta$ or a single $z$.
Whenever the position is labelled by an $\iota$ the automaton checks if the data value is not present in any of the sets and it is then placed in the set $Y_1$ by local updates. Whenever a $\delta$ is read, the data value is ensured to be present only in $Y_1$ and then is moved to $Y_3$ by local updates. In both cases the global update is the identity. Whenever a $z$ is encountered, we have to ensure that $Y_1$ is empty to maintain property (2). This is accomplished by copying the data values in $Y_1$ to $Y_2$ by the global update $\rho=\{(Y_1,Y_2)\}\cup \id_Y$ and adding the current data value to $Y_3$ by local updates. 
Finally, the run is accepted if the sets $Y_1$ and $Y_2$ are empty.

\end{example}

\subsection{Normal Set Automaton}

A set automaton is \emph{normal} if it maintains the following invariant throughout the run on each data word: each data value is present in at most one set.
This is captured by the following syntactic definition.
\begin{definition}[Normal Set Automaton]
A \emph{set automaton} is normal if each of its transitions $(p, \ell, \barz, \rho, \baru, \barv, q)$ satisfy the following properties:
\begin{enumerate}
    \item the set update $\rho$ is a transformation, and,
    \item $\baru$ is a unit vector and if $\baru \neq M^T(\rho)\cdot \barz$ then $\barv = M^T(\rho)\cdot \barz$, otherwise it is $\bar{0}$.
\end{enumerate}   
\end{definition}
In a normal set automaton $\calA$ with sets $Y$, since each global update is a transformation, the update monoid is a transformation monoid, i.e., $M_\calA\leq \bfT_Y$.
Further, for the same reason, data values in each set move to exactly one set. 
This, in combination with restricted local updates where the current data value can be added to at most one set ensures that each data value is present in at most one set throughout the run of a normal set automaton.


For every relation on a set $Y$, there exists a corresponding transformation on the set $\calP(Y)$. 
A consequence of this is that every set automaton can be normalised.

\begin{example}\label{ex:reltotrans}
Consider the relation $\rho$ from \Cref{ex:12IDZ} on the sets $Y=\{Y_1,Y_2,Y_3\}$. 
This relation can be expressed as a transformation $\rho'$ on the sets $\calP(Y)$ where each element represents a subset of elements in $Y$.
For instance, the set $Y_{(100)}$ tracks the data values present in the set $Y_1$ and not in $Y_2$ or $Y_3$, and the set $Y_{(110)}$ tracks the data values present in both $Y_1$ and $Y_2$ but not in $Y_3$.
The relation $\rho$ on $Y$ is captured by the transformation $\rho'$ on $\calP(Y)$ as follows:
$\rho(Y_1)=\{Y_1, Y_2\}$ is given as $\rho'(Y_{(100)})=Y_{(110)}$, $\rho(Y_2)=\{Y_2\}$ as $Y_{(010)}\to \{Y_{(010)}\}$, and $\rho'(Y_3)=Y_3$ as $Y_{(001)}\to \{Y_{(001)}\}$ respectively. 
\end{example}

\begin{propositionrep}\label{prop:sa-nsa}
For each set automaton $\calA$ with the family of sets $Y$ there is an equivalent to a normal set automaton $\calN(\calA)$ with the family of sets $\bftwo^Y\setminus \{\bar{0}\}$.
\end{propositionrep}
\begin{proof}
Consider a set automaton $\calA = (Q, Y, \Sigma, \Delta, I, F, C)$ with the set names $Y=\{1, \ldots, k\}$. 
Let $\calN(\calA)$ be the required set automaton with the set names $\bftwo^Y\setminus \{\bar{0}\}$ where $\bar{0}$ denotes the zero vector.  Automaton $\calN(\calA)$ has the same state space $Q$ as $\calA$ as well as the same sets of initial and final states. 
  The sets of the automaton $\calN(\calA)$  are denoted as $Y'_{\baru}$ for each nonempty vector  $\baru \in \bftwo^Y$.
For each 
\begin{align}
    \delta=(p, \ell, \barz, \rho, \baru, \barv, p') \in \Delta
    \label{dis:trans}
\end{align}
the set of transitions $\Delta'$
of $\calN(\calA)$ contains the tuple
\begin{align}\delta'=(p, \ell, e_{\barz}, \rho', e_{\baru'}, e_{{\barv}'}, p')
    \label{dis:normaltrans}
\end{align}
where 
\begin{align}
\rho' &=\{(\bars,\bars') \mid \bars,\bars'\in \bftwo^Y, \bars' = M^T(\rho) \cdot  \bars\} \\
\baru'&=({M^T(\rho) \cdot \barz} \vee \baru) \wedge \barv^c\\ 
\barv'&=\begin{cases} M^T(\rho) \cdot \barz &\text{if }\baru' \neq M^T(\rho) \cdot \barz\\
\bar{0}&\text{otherwise}
\end{cases}\end{align}
(Recall that $\bars \wedge \bars'$ denote pointwise Boolean conjunction on the vectors $\bars, \bars' \in \bftwo^Y$).
The set of accepting configurations of $\calN(\calA)$ is $C' = \{e_{\baru} \mid \baru \in C\}$. 

We prove the following claim ($\dagger$): 
On a data word $w$, the automaton $\calA$ reaches a configuration $(q,\overline{P})$ if and only if $\calN(\calA)$ reaches the configuration $(q,\overline{Q})$ where $\overline{P}$ and $\overline{Q}$ are mutually determined by each other by the following condition: For each $d\in \calD$, we have that $\bfone_{\overline{P}}(d) = \barz \Longleftrightarrow \bfone_{\barQ}(d) = e_{\barz}$. Clearly, the claim implies that $L(\calA)=L(\calN(\calA))$ by definition of $C'$. 

Next we prove the claim by induction on the length of the data word $w$. When $|w|=0$, the configuration of $\calA$ and $\calN(\calA)$ are initial and the claim holds trivially. Assume the claim holds for all $w$ of length $k$. Assume that $(q,\barP)$ and $(q,\barQ)$ are two corresponding configurations of $\calA$ and $\calN(\calA)$ respectively after reading $w$. 
By induction hypothesis, on the pair $(a,d)\in \Sigma \times \calD$, the transition (\ref{dis:trans}) is applicable in the configuration $(q,\barP)$ if and only if the transition (\ref{dis:normaltrans}) is applicable in the configuration $(q,\barQ)$. Let $(p',\overline{P'})$ and $(p',\overline{Q'})$ be the resulting configurations of $\calA$ and $\calN(\calA)$ respectively. 

We proceed by cases. 
For all data values $d' \in \calD$ such that $d' \neq d$,
\begin{align*}
  \barz =  \bfone_{\overline{P'}}(d')  &\Longleftrightarrow \barz= M^T(\rho) \cdot \bfone_{\barP}(d') \\
    &\Longleftrightarrow e_{\barz}= M^T(\rho') \cdot \bfone_{\barQ}(d') && (\text{By I.H. \& Def.~of $\rho'$)}\\
    &\Longleftrightarrow  e_{\barz} = \bfone_{\overline{Q'}}(d').
\end{align*}
Now for $d'=d$,
\begin{align*}
  \barz =  \bfone_{\overline{P'}}(d)  &\Longleftrightarrow \barz=  (M^T(\rho) \cdot \bfone_{\barP}(d) \vee \baru) \wedge \barv^c \\
    &\Longleftrightarrow e_{\barz} =  (M^T(\rho') \cdot \bfone_{\barQ}(d) \vee e_{\baru'})\wedge e_{\barv'}^c && (\text{By I.H.~ \& Def.~of $\rho'$})\\
    &\Longleftrightarrow  e_{\barz} = \bfone_{\overline{Q'}}(d).
\end{align*}
where $\baru'$ and $\barv'$ are in the definition of $\delta'$. Hence the Claim ($\dagger$) is proved.
\end{proof}

We now generalise normal set automata such that the invariant that each data value is present in at most one set is only maintained in the sets modified by global updates. 

A set $y$ of a set automaton is said to be \emph{stable} if $\rho(y)=\{y\}=\rho^{-1}(y)$ for each global update $\rho$ used in the transitions of the automaton. 

\begin{definition}[Quasi-normal set automaton]
Let $\calA = (Q, Y, \Sigma, \Delta, I, F, C)$ be a set automaton with the set of stable sets $S\subseteq Y$. 
The set automaton $\calA$ is \emph{quasi-normal} if the restriction $\calA\!\restriction_{(Y\setminus S)}=
(Q, Y\setminus S, \Sigma, \Delta', I, F, C')$ of $\calA$ to the non-stable sets is a normal set automaton, given by the following: $(p,\ell,\barz',\rho',\baru',\barv',q) \in \Delta'$  if 
$(p,\ell,\barz,\rho,\baru,\barv,q) \in \Delta$ where 
$\barz'=\barz\!\restriction_{(Y\setminus S)}$, $\rho'=\rho\!\restriction_{(Y\setminus S)}$, $\baru'=\baru\!\restriction_{(Y\setminus S)}$, and $\barv'=\barv\!\restriction_{(Y\setminus S)}$, and $C'=\{\baru\!\restriction_{(Y\setminus S)} \mid \baru \in C\}$.
\end{definition}

%% file: sec3-osbdecidability.tex
\subsection{Ordered Normal Set Automaton} \label{sec:dec}

In this section, we introduce ordered normal set automata and show that they have a decidable emptiness problem by reduction to the emptiness problem of ordered multicounter automata. 
Further, we extend this result to ordered quasi-normal set automata. 

Let $\calA$ be a normal set automaton with the family of sets $Y$.  In a set automaton, new data values are added to a set by local updates. On reading a data word of length $n$, a set of the automaton $\calA$ can have $n$ data values added by $n$ local updates. Thus, in general the contents of a set is unbounded. 
However, this depends on the global updates used in the transitions of $\calA$ allowing for such unbounded accumulation. 

Let $\{\rho_i \mid i \in I\}$ be the set of global updates used by transitions of $\calA$. Let $\rho = \cup_{i\in I}\rho_i$ be the union of the update relations and $\rho^+$ denote the transitive closure of $\rho$.
Consider the graph $(Y, \rho^+)$ and a set $y\in Y$. If there is no vertex $z$ with $(z,z)\in \rho^+$ that can reach $y$, then the induced subgraph consisting of the vertices $Y'\subseteq Y$ that can reach $y$ is acyclic. For each vertex $y' \in Y'$, we can show by induction that if the number of vertices from which $y'$ is reachable is $k$, then $y'$ contains at most $k+1$ data values at each point during the run over a data word. Thus, we have the below definition. 

\begin{definition}[Bounded set]\label{def:boundedset}
    Let $\calA$ be a normal set automaton $\calA$ with the family of sets $Y$. Let $y\in Y$ be a set. In the graph $(Y, \rho^+)$, if no vertex $z$ such that $(z,z) \in \rho^+$ can reach $y$, then the set $y$ is called a \emph{bounded set}. 
    A set $y\in Y$ is called a \emph{bounded set} if in the graph $(Y, \rho^+)$, there is no vertex $z$ such that $(z,z) \in \rho^+$ that can reach $y$. 
\end{definition}



We now introduce ordered normal set automata.
The idea is to impose an order on the non-bounded sets of the automaton and restrict the global updates on the non-bounded sets such that they empty the contents of a prefix of sets and move them to sets higher in the order whose contents not emptied. Such global updates are simulated by an ordered multicounter automaton using its restricted hierarchical zero-tests on a prefix of counters corresponding to the sets emptied by the global update. The contents of the bounded sets are tracked by the ordered multicounter automaton without using zero-tests by maintaining the number of data values in each set in its finite state space. 

Toward this, we recall the notion of an ordered partition.
Let $X$ be a set and $\leq$ be a  linear order on it. For subsets $Y,Y' \subseteq X$, we write $Y \leq Y'$ to mean that $y \leq y'$ for each $y\in Y$ and $y'\in Y'$.
An \emph{ordered partition} of $X$ is a tuple $(X_0, X_1)$ of disjoint subsets that cover $X$ (i.e., $X_0\uplus X_1 = X$) such that $X_0 \leq X_1$.

\begin{definition}[Ordered normal set automaton]
    Let $\calA$ be a normal set automaton with family of sets $Y_\calA$ and bounded sets $Z_\calA\subseteq Y_\calA$. The automaton $\calA$ is  \emph{ordered} if there is a linear order $\leq$ on $Y_\calA$ such that $Y_\calA\setminus Z_\calA \leq Z_\calA$, and for each global update transformation $\rho$ used in the transitions there is an associated ordered partition $(Y_0,Y_1)$ of $Y_\calA\setminus Z_\calA$ such that the tranformation $\rho$ maps the elements of $Y_0$ to $Y_1$, and is the identity transformation on $Y_1$. That is, the contents of each set in $Y_0$ is emptied and moved to a set in $Y_1$ and the contents of each set in $Y_1$ are not emptied. 
    A quasi-normal set automaton $\calB$ with family of sets $Y_\calB$ and stable sets $S_\calB$ is \emph{ordered} if the restriction $\calB\!\restriction_{(Y_\calB\setminus S_\calB)}$ is ordered. 
\end{definition}
Note that there is no restriction on the global update transformations on the bounded sets. 

We now present the main result of this section. 
To this end, we recall the below definition. 


\begin{definition}[Ordered multicounter automaton \cite{Rein08}]
\label{defn:OMA}
An ordered multicounter automaton $\calM$ is a tuple $(Q, \Sigma, O, \Delta, J, F)$ where $Q$ is the finite set of states, $\Sigma$ is the finite alphabet, $O$ is the ordered finite set of non-negative counters, $J$ is the set of initial states, $F\subseteq Q$ is the set of final states. Let $O=\{c_1, \ldots, c_k\}$. The set of transitions is given by
 \[ \Delta \subseteq Q\times (\Sigma \cup \{\varepsilon\})\times \{I_i, D_i, Z_{\leq i} \mid  i\in [k] \}\times Q.\] 
 From a given state $p$ on a label $a$, the transition $(p, a, I_j, q)$ increments the counter $c_j$ by one and moves to the state $q$. Similarly, $D_j$ decrements the counter $c_j$ by one and $Z_{\leq j}$ tests the counters $\{c_i \mid  1\leq i\leq j\}$ for zero value. 
Whenever the run tries to decrement the value of a counter below zero, the computation halts erroneously and rejects.
\end{definition}

A \emph{configuration} of an ordered multicounter automaton is a pair $(q, \bfv)$ where $q\in Q$ and $\bfv = (v_1, \ldots, v_k)$ where $v_i$ is the value of the counter $c_i$. 
An \emph{accepting run} of the automaton over a word is a sequence of transitions from an initial state with all the counters set to zero to a configuration where all counters are zero and the state is final. The language of an ordered multicounter automaton $\calM$ is the set of all finite words $w$ on which $\calM$ has an accepting run. 

The following is a well known result about ordered multicounter automata.

\begin{proposition}[Reinhardt, Theorem 6.1 of \cite{Rein08}]\label{prop:omca-dec}
Emptiness of ordered multicounter automata is decidable. 
\end{proposition}

We simulate the run of an ordered normal set automaton using an ordered multicounter automaton. 
\begin{proposition}\label{prop:osa-to-omca}
For each ordered normal set automaton $\calA$, an ordered multicounter automaton $\calM$ recognising the string projection of $L(\calA)$ can be constructed.
\end{proposition}
\begin{proof}



Let $\calA = (Q, \Sigma, Y, \Delta, I, F, \calC)$ be an ordered normal set automaton with the linear order $\leq$ on the family of sets $Y$. Let $Z \subseteq Y$ be the family of bounded sets of $\calA$. By definition $Y\setminus Z \leq Z$.

We construct an ordered multicounter automaton $\calM$ that simulates the ordered normal set automaton $\calA$. The automaton $\calM$ has a counter $c_y$ for each $y \in Y$ and the counters are ordered with respect to the order $\leq$.
The counter $c_y$ tracks the number of data values in the set $y \in Y$ in a configuration of $\calA$ during the simulation.    The automaton also remembers which of the bounded sets are nonempty and how many data values are in them using its state space. This memory is modified whenever the bounded sets are updated through global or local updates.
This is further elaborated in Item~\ref{item:update} below.

Consider a word $w\in L(\calA)$ with an accepting run $\rho = \delta_1 \delta_2\cdots \delta_n$ of the ordered normal set automaton $\calS$.
Let $\delta=(p, a, e_z, \rho, e_u, \barv, q)$ be a transition in the run. Note that since the automaton $\calA$ is normal, $e_z$ and $e_u$ are unit vectors and $\barv$ is either a unit or zero vector.
It is simulated by the ordered multicounter automaton $\calM$ as follows:

\begin{enumerate}
	\item The states $p, q$ and the current label $a$ are handled by the underlying finite state automaton. 
	\item The automaton $\calM$ tests that the counter $c_z$ is nonzero (by successively decrementing and  incrementing), signifying the presence of a data value with the characteristic vector $e_z$. 
	\item  Let $(Y_0,Y_1)$ be the ordered partition of $Y\setminus Z$  associated with $\rho$. 
Let $y \in Y_0$. The automaton $\calM$ successively decrements the counter $c_y$ and increments the counter $c_{z'}$ such that $e_{z'}=M^T(\rho) \cdot e_y$ (Note that $z' \in Y_1$). 
This process is repeated for each set $y\in Y_0$ (on $\varepsilon$-transitions) until the automaton $\calM$ guesses the counters $\{ c_y \mid y \in Y_0\}$ to be zero.
The guess is verified by zero-testing the set of counters $\{c_{y} \mid y \in Y_0\}$. Note that the sets in $Y_0$ form a prefix that can be zero-tested using the hierarchical zero-tests of $\calM$.
\item \label{item:update} The global update is simulated on the bounded sets by decrementing the counter $c_x$ for $x\in Z$ and incrementing the counter $c_{x'}$ such that $e_{x'}=M^T(\rho) \cdot e_x$ (Note that $x' \in Y$ and not necessarily in $Z$), until the former become zero. This step is accomplished without the use of zero-tests by remembering in the state space the size of each of the bounded sets that are nonempty. Further, the finite state automaton remembers this information after each transition.
	\item Finally the local updates $e_u$ and $\barv$ are simulated by incrementing the counter $c_u$ and decrementing the counter $c_v$ if $\barv\neq \overline{0}$. 
    
\end{enumerate}

The transitions $\delta_1$ and $\delta_n$ are ensured to be initial and final transitions, respectively, by the underlying finite state automaton. At the end of the run decrements all counters $c_x$ such that $e_x \in C$ to zero and performs a zero-test on all the counters verifying that all data values reached an accepting configuration.
\end{proof}

Thus, by \cref{prop:omca-dec,prop:osa-to-omca} we obtain the below result. 
\begin{theorem}
Emptiness of ordered normal set automata is decidable. 
\end{theorem}

We now extend the above decidability result to ordered quasi-normal set automata. 

\begin{propositionrep}
    For each ordered quasi-normal set automaton there is an equivalent ordered normal set automaton.
\end{propositionrep}
\begin{proofsketch}
    Given an ordered quasi-normal set automaton $\calA$, consider the set automaton $\calN(\calA)$ obtained on normalising it as given by \Cref{prop:sa-nsa}. 
    We show that the restriction of $\calN(\calA)$ to the family of sets given by the unit vectors corresponding to the non-stable sets of $\calA$, is an equivalent ordered set automaton. 
\end{proofsketch}
\begin{proof}
Assume that $\calA = (Q, Y, \Sigma, \Delta, I, F, C)$ is an ordered quasi-normal set automaton with the family of stable sets $S\subseteq Y$. By assumption the set automaton $\calA\!\restriction_{Y \setminus S}$ is ordered by a linear order $\leq$ on $Y\setminus S$. Let $U = \{\bar0 \} \cup \{e_y \mid y \in Y\setminus S\}$ be a set of characteristic vectors. Let $\calN(\calA)$ be the normal set automaton obtained on normalising $\calA$ as given in \Cref{prop:sa-nsa}. Consider the restriction $\calN(\calA)\!\restriction_{Y'}$ with the family of sets
$Y'=\{\baru \in  \bftwo^Y 
\mid \baru\!\restriction_{(Y\setminus S)} \in U\}$ where the restriction of the vectors to the non-stable sets of $Y\setminus S$ of $\calA$ are unit vectors, i.e., the non-stable sets occur at most once. Let $\leq'$  be a linear order on $Y'$ such that $\leq' \,\supseteq\, \{(\baru, \barv) \mid \baru\!\restriction_{Y\setminus S} = e_u, \baru\!\restriction_{Y\setminus S} = e_v, u \leq v\}$. It is easy to verify that 
$\calN(\calA)\!\restriction_{Y'}$ is an ordered set automaton with the linear order $\leq'$ that is equivalent to $\calA$.
\end{proof}
\begin{corollary}\label{cor:oqnsa-dec}
Emptiness of ordered quasi-normal set automata is decidable.
\end{corollary}

\subsection{Closure Properties}

Closure properties of set automata follow from that of class automata (See \Cref{prop:sa=ca}). 
We prove closure properties for quasi-normal set automata required for the translation of logical formulas to automata. 

\begin{lemmarep}\label{lem:qnsa-closure}
Quasi-normal set automata are closed under letter-to-letter renaming, and union and intersection with data automata. 
\end{lemmarep}
\begin{proofsketch}
    It is easy to verify that the standard construction for closure under letter-to-letter renaming preserves quasi-normality. 
    For closure under union and intersection with data automata, we use our result that each data automata is equivalent to a set automaton where each set is stable (See \Cref{prop:isa=cma}). The result follows from the standard constructions using this equivalent set automaton.
\end{proofsketch}
\begin{proof}


Consider a quasi-normal set automaton $\calA$ with sets $Y_\calA$ and stable sets $S_\calA$ and a data automaton $\calB$ with sets $Y_\calB$ over the alphabet $\Sigma$. By \Cref{prop:isa=cma}, for the data automaton $\calB$, there exists an equivalent set automaton $\calB'$ with sets $S_{\calB'}$ where each set is stable. 

If $h:\Sigma\rightarrow \Sigma'$ is a renaming from alphabet $\Sigma$ to an alphabet $\Sigma'$, it is easily checked that $h(L(\calA))$ is accepted by the automaton $\calA'$ obtained by applying the renaming $h$ to each of the transitions in $\calA$. Further, this construction preserves quasi-normality of the set automaton.

Closure under union can be shown by taking the disjoint union of $\calA$ and $\calB'$. The resulting set automaton has the family of stable sets $S_\calA\cup S_{\calB'}$.
Since all the sets of $\calB'$ are stable, the restriction of the resulting set automaton to the non-stable sets $Y_\calA\setminus S_\calA$ results in a normal set automaton due to the quasi-normality of $\calA$. 

For closure under intersection, consider the synchronous product of $\calA$ and $\calB'$. By standard arguments, this automaton accepts the language $L(\auto_1) \cap L(\auto_2)$. Similarly, here the resulting automaton is quasi-normal. 
\end{proof}

%% file: sec4-logic.tex
\section{Equivalence of Set Automata and $\FO^2$ with Guarded Regular Predicates}\label{sec:logic}
\nosectionappendix

\begin{toappendix}
   \subsection*{Proofs for \Cref{sec:logic}}
\end{toappendix}

In this section, we present a logical characterisation of set automata. 

\subsection{Set Automata to $\FO^2$ with Guarded Regular Predicates}


\begin{propositionrep}\label{prop:aut-to-logic}
For each set automaton $\calA$ with update monoid $M$ there is a formula $\phi_\calA \in \EMSO^2[\Sigma, <,  +1,  \sim, \moo1, \widetilde{\calF}_{M}]$ such that $L(\calA) = L(\phi_\calA)$.  
\end{propositionrep}
\begin{proofsketch}
We write a formula that is satisfied by a data word $w=(a_1,d_1)\cdots(a_n,d_n)$ if and only if there is a successful run of $\calA$ on $w$ using standard ideas. 
In this sketch, we focus on how the global and local updates in the transitions of $\calA$ are captured by a formula with guarded regular predicates. 

We ensure that (1) at each position the sets in which the input data value is present is consistent with the global and local updates of the transition, and (2) in two class-successive positions, the sets in which the data value is present is consistent with the global updates in the positions strictly between them. 
Note that we do not need to consider the local updates in the positions between two successive positions of a given class as local updates only affect the input data value at that position.


Let $\Delta$ and $\calS$ denote the set of transitions and family of sets of $\calA$ respectively. 
We use the monadic predicates $\overline{X}=\{X_\delta \mid \delta \in \Delta\}$ to indicate the transition used at a position in the run. We use the monadic predicates $\overline{Y} = \{Y_s \mid s \in \calS\}$ to indicate the contents of the sets during the run: $Y_s$ is true, for $s\in \calS$, at a position $j$ if during the run of $\calA$ after reading $(a_j,d_j)$ we have that $d_j$ is present in the set $Y_s$. Let $\barZ=\barX\cup\barY$. 
Our formula is of the form $\exists\barZ  \,\varphi$ where $\varphi \in \FO^2[\Sigma,\barZ, <, +1,  \sim, \moo1, \widetilde{\calF}_M]$ is the conjunction of the formulas detailed in this sketch.
To capture that the characteristic vector of membership at position $x$ is $\rho\in \bftwo^{\calS}$, we write,
\[\rho(x) \coloneqq \bigwedge_{\substack{s \in \calS,\\  \rho(s) = 1}} Y_s(x) \wedge \bigwedge_{\substack{s \in \calS,\\ \rho(s) = 0}} \neg Y_s(x).\]

First, we ensure that the updates made at each position $x$ with label $(a_i,d_i)$ is consistent with the transition $\delta=(p,a_i,\chi,t,\alpha,\beta,q)$ at position $x$.
That is, the current data value $d_i$ with characteristic vector $\chi$ after the application of the transition $\delta$ with global update $t$ and local updates $\alpha$ and $\beta$ has the characteristic vector of membership $\rho_\delta=M^T(t)\cdot \chi +\alpha-\beta$, i.e., $\rho(x)$ is true. We write
\[
\forall x\ ( X_\delta(x) \rightarrow \rho_\delta(x)).
\]

Now, it suffices to ensure that the contents of the sets in two class-successive positions are consistent with the global updates in positions strictly between them. Toward this, we introduce the below notation.
Let $\bftwo^{\barZ}$ denote the characteristic vectors of the monadic predicates $\barZ$. 
Let $\Sigma'=\Sigma\times \bftwo^{\barZ}$.

Let $\Sigma' \subseteq \Sigma \times {\bftwo^{\overline{X}}}$ be the set of all vectors that denotes the interpretation of the monadic variables $\overline{Y}$. 
Let $\pi_{\Delta}:\Sigma'\rightarrow \Delta$ be the projection function that maps a letter $\ell \in \Sigma'$ to the (unique) $\delta \in \Delta$ such that $X_\delta$ is true in $\ell$.  We extend the map $\pi_{\Delta}$ to the map  $\pi_{\Delta}^*: \Sigma'^* \rightarrow \Delta^*$ in the natural manner.

Let $\pi:\Delta \rightarrow M$ be the projection map defined as $\delta = (p,\gamma, \chi, t, \alpha, \beta, q) \mapsto t$ that maps each transition to its transformation given as an element of the update monoid $M$. 
We extend the map $\pi$ to the map $\pi^*:\Delta^*\rightarrow M^*$ in the natural manner.

Let $\sigma:M^* \rightarrow M$ be the product morphism that maps each word $m_1\cdots m_n$ to the product of $m_1,\ldots, m_n$.
Clearly, the morphism $\sigma \circ \pi^* \circ \pi_{\Delta}^* : \Sigma'^* \rightarrow M$ maps a sequence of letters in $\Sigma'$ to the product of the corresponding transformations in $M$.


For $\chi \in \bftwo^{\overline{\calS}}$, let $\Delta_\chi \subseteq \Delta$ denote all transitions with the membership vector $\chi$. 
Assume that a position $y$ is labelled with a transition $\delta = (p, a_j, \chi, t, \alpha, \beta, p') \in \Delta$. 
Let position $y$ be the class-successor of position $x$. Assume that the position $x$ satisfies $\rho(x)$, i.e., the characteristic vector of membership of the data value in position $x$ is $\rho$, and $y$ is labelled by a transition in $\Delta_\chi$. Then we ensure that the product of transformations between the positions strictly in between $x$ and $y$ is an element $m\in M$ such that $\chi = \rho \cdot m$. We write

\[ \forall x \forall y \bigwedge_{\substack{\chi,\rho\in \bftwo^{\calS}\\ \delta\in \Delta_\chi}}(x\moo1 = y \,\wedge\, \rho(x) \,\wedge\, X_\delta(y) \rightarrow {\widetilde{L}}(x,y))
\] %
where $L \coloneqq \bigcup_{\substack{m \in M, \\ \chi = \rho \cdot m}} (\sigma \circ \pi^* \circ \pi_{\Delta}^*)^{-1}(m)$.
\end{proofsketch}
\begin{proof}

Let $\calA$ be a set automaton with update monoid $M$. 

We write a formula that is satisfied by a data word $w=(a_1,d_1)\cdots(a_n,d_n)$ if and only if there is a successful run of $\calA$ on $w$. 

Let $\Delta$ be the set of transitions of $\calA$. 
Let $\overline{X}$ be the set of monadic predicates $\{X_\delta \mid \delta \in \Delta\}$ which are used to indicate the transition used at a position in the run.
Let $\calS$ denote the family of sets of $\calA$. We use the monadic predicates $\overline{Y} = \{Y_s \mid s \in \calS\}$ to indicate the contents of the sets during the run: $Y_s$ is true, for $s\in \calS$, at a position $j$ if during the run of $\calA$ after reading $(a_j,d_j)$ we have that $d_j$ is present in the set $Y_s$. Let $\barZ=\barX\cup\barY$.

Our formula is of the form $\exists\barZ  \,\varphi$ where $\varphi \in \FO^2[\Sigma,\barZ, <, +1,  \sim, \moo1, \widetilde{\calF}_M]$ is the conjunction of the following formulas.

\begin{enumerate}

\item\label{item:uniquetransition} Each position is labelled by exactly one transition of the automaton, i.e.,
\[
\forall x \bigvee_{\delta\in \Delta} X_\delta(x)  \,\wedge\,  
\forall x   \bigwedge_{\substack{\delta,\delta'\in \Delta\\ \delta \neq \delta'}}\neg (X_\delta(x) \wedge X_{\delta'}(x)).
\]

\item Let $I \subseteq \Delta$ denote the set of transitions from an initial state of the automaton. 
We ensure that the first position is labelled by transitions from $I$, i.e., 
\[
\forall x  \,(\forall y\ y+1 \neq x  \rightarrow \bigvee_{\delta\in I} X_\delta(x)). 
\]

\item  Let $F \subseteq \Delta$ denote the set of transitions to a final state of the automaton. 
We ensure that the last position is labelled by a transition in $F$, i.e.,
\[
\forall x  \,(\forall y\ x+1 \neq y  \rightarrow \bigvee_{\delta\in F} X_\delta(x)). 
\]

\item We say a pair of transitions $(\delta,\delta')$ are \emph{compatible}, denoted as $\mathit{comp}(\delta,\delta')$, if the target state of $\delta$ is the source state of $\delta'$. We state that consecutive positions are labelled by compatible transitions, i.e., 
\[
\forall x \forall y \,(x +1 = y  \rightarrow \bigvee_{\substack{\delta, \delta'\in\Delta,\\ \mathit{comp}(\delta,\delta')}} X_\delta(x) \wedge X_{\delta'}(y)).
\]

\item Let $\Delta_{\bar{0}}$ be the set of transitions where the characteristic vector is the zero-vector. We ensure that class-minimal positions are labelled by transitions from $\Delta_{\bar{0}}$, i.e., 
\[ 
\forall x  \,( \neg \exists y\, y \moo 1 = x  \rightarrow \bigvee_{\delta\in\Delta_{\bar{0}}}  X_\delta(x)). 
\]

\item Assume that the set automaton $\calA$ has the acceptance condition (2) given in \Cref{lem:acceptance}. Let $F \subseteq \Delta$ denote the set of transitions to a local final state of $\calA$. We state that every class-maximal position is labelled by a transition to a local final state, i.e., 
\[
\forall x  \,( \forall y\,  x \moo 1 \neq y \rightarrow \bigvee_{\delta\in F} X_\delta(x)).
\]

\item It remains to ensure that the data values present in the sets $\calS$ are according to the updates made by the transitions in $\Delta$. That is, the labelling by the $\overline{Y}$-predicates are consistent with the run indicated by the $\overline{X}$-predicates. Toward this, we introduce the below formula. 
Let $\rho \in \bftwo^{\calS}$ be a characteristic vector. 
To capture that the characteristic vector of membership at position $x$ is $\rho$, we write,
\[\rho(x) \coloneqq \bigwedge_{\substack{s \in \calS,\\  \rho(s) = 1}} Y_s(x) \wedge \bigwedge_{\substack{s \in \calS,\\ \rho(s) = 0}} \neg Y_s(x).\]

Now, we ensure that the updates made at each position $x$ with label $(a_i,d_i)$ is consistent with the unique transition $\delta=(p,a_i,\chi,t,\alpha,\beta,q)$ at position $x$.
That is, the current data value $d_i$ with characteristic vector $\chi$ after the application of the transition $\delta$ with global update $t$ and local updates $\alpha$ and $\beta$ has the characteristic vector of membership $\rho_\delta=M^T(t)\cdot \chi +\alpha-\beta$. We write
\[
\forall x\ ( X_\delta(x) \rightarrow \rho_\delta(x)).
\]

\item Finally, we ensure that the data values present in the sets $\calS$ in two class-successive positions are consistent with the global updates in the positions strictly between them. 

Let $\bftwo^{\barZ}$ denote the characteristic vectors of the monadic predicates $\barZ$. 
Let $\Sigma'=\Sigma\times \bftwo^{\barZ}$.
Let $\Sigma' \subseteq \Sigma \times {\bftwo^{\overline{X}}}$ be the set of all vectors that denotes the interpretation of the monadic variables $\overline{Y}$. 

Let $\pi_{\Delta}:\Sigma'\rightarrow \Delta$ be the projection function that maps a letter $\ell \in \Sigma'$ to the unique $\delta \in \Delta$ determined by Item~\ref{item:uniquetransition} above such that $X_\delta$ is true in $\ell$.  We extend the map $\pi_{\Delta}$ to the map  $\pi_{\Delta}^*: \Sigma'^* \rightarrow \Delta^*$ in the natural manner.

Let $\pi:\Delta \rightarrow M$ be the projection map defined as $\delta = (p,\gamma, \chi, t, \alpha, \beta, q) \mapsto t$ that maps each transition to its transformation given as an element of the update monoid $M$. 
We extend the map $\pi$ to the map $\pi^*:\Delta^*\rightarrow M^*$ in the natural manner.

Let $\sigma:M^* \rightarrow M$ be the product morphism that maps each word $m_1\cdots m_n$ to the product of $m_1,\ldots, m_n$.
Clearly, the morphism $\sigma \circ \pi^* \circ \pi_{\Delta}^* : \Sigma'^* \rightarrow M$ maps a sequence of letters in $\Sigma'$ to the product of the corresponding transformations in $M$.

For $\chi \in \bftwo^{\overline{\calS}}$, let $\Delta_\chi \subseteq \Delta$ denote all transitions with the membership vector $\chi$. 
Assume that a position $x$ is labelled with a transition $\delta = (p, a_j, \chi, t, \alpha, \beta, p') \in \Delta$. 
Let position $y$ be the class-successor of position $x$. Assume that the position $x$ satisfies $\rho(x)$, i.e., the characteristic vector of membership of the data value in position $x$ is $\rho$, and $y$ is labelled by a transition in $\Delta_\chi$. Then we ensure that the product of transformations between the positions strictly in between $x$ and $y$ is an element $m\in M$ such that $\chi = \rho \cdot m$. We write

\[ \forall x \forall y \bigwedge_{\substack{\chi,\rho\in \bftwo^{\calS}\\ \delta\in \Delta_\chi}}(x\moo1 = y \,\wedge\, \rho(x) \,\wedge\, X_\delta(y) \rightarrow {\widetilde{L}}(x,y))
\]
where $L \coloneqq \bigcup_{\substack{m \in M, \\ \chi = \rho \cdot m}} (\sigma \circ \pi^* \circ \pi_{\Delta}^*)^{-1}(m)$.
\end{enumerate}
The proof is completed by examining that a data word $w$ satisfies the formula $\exists \barZ \varphi$ if and only if $\calA$ has a successful run on $w$.
%
\end{proof}


\subsection{$\FO^2$ with Guarded Regular Predicates to Quasi-Normal Set Automata}\label{subsec:guardedlogic-to-automata}

In this section we translate the guarded logic formulas to a subclass of quasi-normal set automata called \emph{suffix-storing set automata}. 

\begin{definition}[Suffix-storing set automaton]
Let $M$ be a monoid and $h:\Sigma^* \to M$ be a morphism. A set automaton over the alphabet $\Sigma$ is \emph{$h$-suffix-storing} if it is a quasi-normal set automaton whose non-stable family of sets is of the form $\{X_m \mid m \in M\}$ and that obeys the following property: after reading the input prefix $w \in (\Sigma\times \calD)^*$, the set $\{ m \mid X_m \neq \emptyset\}$ is a subset of images of suffixes of $\str(w)$ under $h$.
\end{definition}

The rest of the section is devoted to the proof of \Cref{prop:logic-to-aut}.

\begin{proposition}\label{prop:logic-to-aut}
For each formula $\psi \in \EMSO^2[\Sigma, <, +1,  \sim, \moo1, \widetilde{\calF}_h]$, there is a $h$-suffix-storing set automaton $\calA_\psi$ such that $L(\psi)=L(\calA_\psi)$.  
\end{proposition}
We proceed by translating the formulas to suffix-storing set automata.

Assume that we are given a formula $\psi$ of the form $\exists \barX \varphi$ where $\varphi\in\FO^2[\Sigma,\barX, <,  +1,  \sim, \moo1, \widetilde{\calF}_h]$ and $\barX=\{X_1,\ldots,X_m\}$. 
Further, assume that $\varphi$ uses guarded predicates defined by the morphism $h: (\Sigma \times  2^{\overline{X}})^* \rightarrow M$. Since set automata are closed under letter-to-letter renaming  it suffices to show that there is a set automaton $\calA_\varphi$ recognising $L(\varphi)$. In the following, to transform the formulas, we add additional monadic variables to the vocabulary. This requires $h$ to be replaced by the morphism $h \circ \pi$ where $\pi$ is the projection map to the alphabet $\Sigma \times  2^{\overline{X}}$. To keep the discussion simple, we omit referring to the step explicitly.

We convert the formula $\varphi$ to an equivalent formula of size linear in $|\varphi|$ with additional unary predicates $\overline{Y}=\{Y_{1}, \ldots, Y_n\}$ in Scott Normal Form (see, for instance \cite{scottnormalform}). 
Hence, we get an equivalent formula 
\begin{align}
\exists Y_{1} \cdots Y_n \left(\forall x \forall y\, \chi \land \bigwedge\limits_i \forall x \exists y\, \chi_i\right)
\end{align}
where $Y_{1}, \ldots, Y_n$ are new unary predicates and $\chi$ and $\chi_i$ are quantifier-free with free variables $\overline{X}\cup \overline{Y}$. Again, by appealing to the closure under letter-to-letter renaming of set automata, it suffices to show that there is an equivalent set automaton for formulas of the form 
\begin{align}\label{eq:logictoaut-formula}
    \forall x \forall y\, \chi \land \bigwedge\limits_i \forall x \exists y\, \chi_i.
\end{align}

\paragraph*{Unary and Binary types}\label{pf:unarybinarytypes}

Let $x\ll y$ stand for the formula $x<y\land x+1\neq y$, i.e., $x$ occurs before $y$ but is not its predecessor.  Let $O$ denote the set of binary \emph{order-types} on two variables, i.e., 
\[ O=\{x \ll y, x+1=y, x = y, y+1=x, y \ll x\}.\]
No two formulas in $O$ are satisfiable at the same time, and any quantifier-free formula that uses only the predicates $\{<, +1,=\}$ is equivalent to a disjunction of formulas from $O$ (easily verified by converting the formula to DNF).

Let $x \farsim y$, read as `$x$ and $y$ are class-distant', stand for the formula $x \sim y \wedge x \moo1 \neq  y \wedge x \moo1 \neq y$, i.e., $x$ and $y$ are in the same class but neither is the class successor of the other.
Let $E$ denote the set of binary \emph{equivalence-types} on two variables, i.e., 
\[ E=\{ x \not \sim y, x \farsim y, x\moo1=y,  y\moo1=x, x=y\}.\]
It is easily verified that no two formulas in $E$ are satisfiable at the same time and that any quantifier-free two-variable formula using the predicates $\{ \sim, \moo1\}$ is equivalent to a disjunction of formulas in $E$.

A \emph{literal} is either an atomic formula or a negation of an atomic formula. Let $\Gamma$ be a set of literals. The set $\Gamma$ is \emph{consistent} if $\Gamma$ does not contain a literal and its negation. It is \emph{maximally-consistent} if $\Gamma$ is consistent and no strict superset of $\Gamma$ is consistent.

A \emph{unary-type} over $\Sigma \cup \barX$ is a conjunction of a positive literal from $\Sigma$  and a maximally-consistent set of literals over $\barX$ on the same variable (either $x$ or $y$). For example, let $\Sigma = \{a,b\}$. Then  $a(x) \wedge \neg X_1(x)$ is a unary-type over the unary predicates $\{a,b,X_1\}$, whereas $a(x)$  and $a(x) \wedge b(y) \wedge X_1(y)$ are not. Let $U$ denote the set of unary-types. 



\paragraph*{Construction of the Quasi-Normal Set Automaton}
We first transform the formulas $\forall x \forall y\, \chi$ and $\forall x \exists y\, \chi_i$ into suitable forms given below using standard techniques (see Lemma 12 and 13 of \cite{BMSSD11}).

It is straightforward to first transform $\chi$ to CNF and distribute the universal quantifiers over the conjunction. 
Furthermore, by using standard logical equivalences, we write each of the exponentially many conjuncts in $\chi$ obtained from the CNF transformation in the below form: 
\begin{align}\label{dis:AA}
\phi&=\forall x \forall y (\alpha(x)\land\beta(y)\land o(x,y)\land \epsilon(x,y) \to  \widetilde\lambda(x,y))
\end{align}

where $\alpha, \beta \in U$, $o(x,y) \in O$, $\epsilon(x,y) \in E$, and $\widetilde\lambda(x,y)$ is either $\mathit{false}$ or a regular predicate of the form $\tilL(x,y)$ or $\tilL(y,x)$ for some language $L$ defined by $h$. 
The set automaton constructed ensures that for each $\alpha$ and $\beta$ with the class and order conditions met, the factor bordered by them satisfies the guarded predicate $\widetilde\lambda$ by storing information about the factors in its sets. 

Similarly, for each $\chi_i$, by introducing additional unary predicates, it suffices to consider formulas of the form 
\begin{align}
\phi&=\forall x \exists y (\alpha(x) \to \beta(y)\land o(x,y)\land \epsilon(x,y) \land \widetilde\lambda(x,y)).
\end{align}

Here, the set automaton ensures that each $\alpha$ is witnessed by a $\beta$ satisfying the order, class, and guarded predicate constraints. 
Thus, it suffices to construct a set automaton for a formula with conjuncts in either of the above forms.  
Let $k$ be the number of such conjuncts. 

For the construction of the set automaton, we use the sets $X = \{X_m \mid m \in M\}, A = \{A_\gamma\mid \gamma\in U\}$ and a collection of sets $Y_i$ for each of the $k$ conjuncts.
We refer to the sets together as $S=X\cup A\cup \bigcup\limits_{i\in[k]} Y_i$.

The sets in $X$ are used to compute the image of factors (suffixes) under the morphism $h$. 
The sets in $A$ are used to track which unary-types correspond to the factors whose images are computed in $X$. 
The sets in $Y_i$ are used to ensure that the $i$'th conjunct is satisfied during the run. 
The sets in $A$ and $Y_i$ for each $i\in [k]$ are stable and only the identity transformation is applied to them. 

The automaton uses the set update transformations $T= \{t_\gamma \mid \gamma \in U\}$ where the transformation $t_\gamma$ is given by $X_p\mapsto \{X_{p\cdot h(\gamma)}\}$ for each $p\in M$ and $Z\mapsto\{Z\}$ for $Z\in S\setminus X$.  






Now, it suffices to prove \Cref{lem:guardedlogic-construction}.

\begin{lemma}
\label{lem:guardedlogic-construction}
For each formula $=\bigwedge\limits_{i\in[k]}\phi_i$ where $\phi_i$ is a formula of the form
\begin{align*}  
(\forall\forall):\ \ & \forall x \forall y\ (\alpha(x)\land\beta(y)\land o(x,y)\land \epsilon(x,y) \to  \widetilde\lambda(x,y)) \mathrm{\ or\ } \\
(\forall\exists):\ \ & \forall x \exists y\ (\alpha(x) \to \beta(y)\land o(x,y)\land \epsilon(x,y) \land \widetilde\lambda(x,y))
\end{align*}
there is an equivalent $h$-suffix-storing set automaton where, in addition to the identity transformation on the sets, the set updates in $T$ are used.
\end{lemma}    

We construct a set automaton $\calA$ with the sets $S$. 
First we introduce some notation used in the proof. 

\smallskip
\noindent\textbf{Notation:}
For a subset $I$ of positions, let $\gamma I$ denote the set of all positions in $I$ that are labelled by $\gamma$. 
For example $\gamma(0,\infty)=\gamma[1,\infty)$ is the set of all positions in the word labelled by $\gamma$.
Similarly, for a data value $d$, let $d I$ denote the set of all positions in $I$ that are in the class of $d$.
Further, let $\gamma^d I$ denote the set of all positions in $I$ that are in the class of $d$ and is labelled by $\gamma$. 
Let $\hat{d}(x)$ denote the largest position in $d(0,x)$ if it exists, i.e., the latest occurrence of data value $d$ that is smaller than $x$. We refer to this position as the \emph{latest} occurrence of $d$ (with respect to the current position).  
Let $\hat{\gamma}^d(x)$  denote the largest position in $\gamma^d(0,x)$ if it exists, i.e., the position of latest occurrence of $\gamma$ in the class of $d$ that is smaller than $x$. We refer to this position as the \emph{latest} $\gamma$ of $d$ (with respect to the current position).  All the other positions in 
 $\gamma^d(0,x)$ are referred to as \emph{non-latest} $\gamma$ of $d$. 
For a word $w$ and an interval $I$ of positions, let $w_I$  denote the factor of $w$ given by the positions in $I$. 
For example $w_{(i,i+1)}$ is always the empty word whereas $w_{(i-1,i]}$ is the $i$-th letter. 
For a word $w$ and a pair of positions $i,j$, the factor $w_{(i,j)}$ is called the factor \emph{bordered} by the positions $i$ and $j$.

In the required set automaton $\calA$, the sets in $X$ are used to store the image of factors (suffixes) of the input word in the monoid $M$ under the morphism $h$. This is accomplished by storing the data value of the first position of the factor in a set $X_m$ if the image of the factor is $m\in M$.
This can be thought of as tracking the run of the monoid on the factors. 
Thus, in the rest of the proof, we say that the \emph{runs of the monoid are stored in the sets} to refer to the images of the factors that are stored.
It is sufficient for $X$ to store the images of factors under the morphism $h$ as all the languages used in the guarded predicates are recognised by it. 

The required set automaton $\calA$ uses the sets in $Y_i$ to ensure that the formula $\phi_i$ is satisfied during the run. 
The state space of the automaton 
is given by the cartesian product of the set of states obtained by the construction for each of the formulas $\phi_i$ as described below, with the transitions and the state space ensuring that each of the formulas $\phi_i$ are satisfied. 
The acceptance conditions for each of the formulas $\phi_i$ as described in the construction below are taken in conjunction at the end of a run, with the automaton accepting only if the acceptance conditions according to the construction for each formula $\phi_i$ are satisfied. 

We perform a case analysis on the types $o(x,y)$, $\epsilon(x,y)$, and $\widetilde\lambda(x,y)$ for each of the forms $\phi_i$ occurs in. 
We denote the unary types occurring in $\phi_i$ as $\alpha_i$ and $\beta_i$, respectively. Further, let $L_i$ be the language used in the guarded predicate of the form $\tilL(x,y)$ recognised by morphism $h$ with accepting set $P_i\subseteq M$.

We give the construction by first describing how the sets in $X$ and $A$ track the runs of the monoid. This is followed by describing how the cases for each $\phi_i$ are handled using the sets in $Y_i$ along with the run information from $X$ and $A$. 

The sets in $X$ are used to compute the images of factors defined by the latest occurrence of a data value and the current position. 
The sets in $A$ are used to remember the label present in the latest occurrence of a data value. 
The set automaton $\calA$ reads the data word from left to right. Let $w$ be the string projection of the data word.
Assume that the automaton is on a position $y$. 
Let $\alpha_i$ denote the unary type $\alpha$ in the conjunct $\phi_i$.
Let $d$ be a data value and let $x_0$ be the latest occurrence of $d$ before the current position if it exists, i.e., $x_0=\hat{d}(y)$.
The following invariants are maintained during the run of the set automaton $\calA$: 
\begin{itemize}
    \item[(X)] The data value $d$ is in the set $X_m$, where $m\in M$, if and only if the factor bordered by latest occurrence of $d$ and the current position has image $m\in M$. That is, the image of $w_{(x_0,y)}$ is $m$. Note that a data value can be present in at most one of the sets in $X$. 
    \item[(A)] The data value $d$ is in the set $A_{\gamma}$ if and only if the latest occurrence of $d$ before $y$ was with the label $\gamma$. That is, $x_0=\hat{\gamma}^d(y)$.
\end{itemize}

On encountering a pair $(\ell, d')$, the automaton performs the following actions. 
The set update $t_\ell$ is performed, updating the \emph{runs} stored in the sets. Since all the guarded predicates' languages are recognised by the same morphism $h$, their runs can be tracked simultaneously. 
The data value $d'$ is removed from the unique set in $X$ in which it is present, if any. 
Further, to maintain the invariant (X) for the upcoming positions in the run, the data value $d'$ is added to the set $X_1$ (here, $1$ is the identity of the monoid $M$) by means of local updates since the run of the monoid $M$ starts at the identity element. 
To maintain the invariant (A), the data value $d'$ is removed from the sets in (A) in which it is present, if any, and is added to the set $A_\ell$. 
Note that the sets in $X$ obey the suffix-storing property.

We now obtain \Cref{prop:logic-to-aut} from \Cref{lem:AA-logic-to-aut,lem:AE-logic-to-aut}. We omit their proofs here due to space constraints. 
\begin{lemmarep}\label{lem:AA-logic-to-aut}
    
    For each formula $\phi_i$ of the form
    \[(\forall\forall):\ \  \forall x \forall y\ (\alpha(x)\land\beta(y)\land o(x,y)\land \epsilon(x,y) \to  \widetilde\lambda(x,y)) 
    \]
    the $h$-suffix-storing set automaton $\calA$ can be constructed such that, using the sets $X\cup A\cup Y_i$ it ensures that a data word is accepted by the automaton $\calA$ if and only if it satisfies the formula $\phi_i$. 
\end{lemmarep}
\begin{proof}
We first discuss the degenerate cases below. 

\smallskip
\noindent\textbf{When $\widetilde\lambda(x,y)$ is $\mathit{false}$}: In this case there is an equivalent data automaton for the formula $\phi_i$ since it uses no regular predicates (see \cite{BMSSD11}). 
Since data automata are captured by set automata with identity transformations (see \Cref{prop:isa=cma}), we simulate it using local updates in the sets $Y_i$. 
Henceforth we assume that $\widetilde\lambda(x,y)$ is a guarded predicate of the form $\widetilde{L}_i(x,y)$ or $\widetilde{L}_i(y,x)$.

\smallskip
\noindent\textbf{When $o(x,y)$ is $x=y$}: If $\epsilon(x,y) \in \{  x\moo1 =y,  y\moo1 = x, x\not \sim y, x \farsim y\}$, or if $\alpha_i \neq \beta_i$, then the premise of $\phi_i$ is false. Thus $\phi_i$ is always true and the required set automaton accepts all inputs. 
Next assume that $\epsilon(x,y)$ is $x = y$ and $\alpha_i = \beta_i$. Then the conclusion $\widetilde{L}_i(x,y)$ is false as per the semantics of the guarded predicates since $x=y$. Thus, for the formula to be true, the premise should not be met. 
This can be ensured by a data automaton as the premise does not contain any regular predicates (See \cite{BMSSD11}). 
Thus we take the corresponding set automaton with identity transformations on the sets $Y_i$ (See \Cref{prop:isa=cma}). 
Henceforth we assume that $o(x,y)$ is not $x=y$.

\smallskip
\noindent\textbf{When $o(x,y)$ is $x\ll y$ or $x+1=y$ and $\widetilde\lambda(x,y)$ is $\widetilde{L}_i(y,x)$}:
Here, the conclusion is always false due to the semantics of the guarded predicates. 
Thus, for $\phi_i$ to be true, we ensure that the premise is not met using a set automaton with identity transformations.
A symmetric case occurs when $o(x,y)$ is $y\ll x$ or $y+1=x$ and $\widetilde\lambda(x,y)$ is $\widetilde{L}_i(x,y)$.
Henceforth, whenever $o(x,y)$ is $x\ll y$ or $x+1=y$, we assume that $\widetilde\lambda(x,y)$ is $\widetilde{L}_i(x,y)$ and whenever $o(x,y)$ is $y\ll x$ or $y+1=x$, we assume $\widetilde\lambda(x,y)$ to be $\widetilde{L}_i(y,x)$.

\smallskip
\noindent\textbf{When $\epsilon(x,y)$ is $x\nsim y$ and $o(x,y)$ is not $x=y$}: Here, if the premise is met, the conclusion can never be met as $\widetilde{L}_i(x,y)$ requires that $x\sim y$. Thus, for $\phi_i$ to be true, we ensure that the premise is not met using a set automaton with identity transformations.

\smallskip 
\noindent Consider the following cases, (1) $o(x,y)$ is $x\ll y$ or $x+1=y$ and $\epsilon(x,y)$ is  $y \moo1 = x$, (2) $o(x,y)$ is $y\ll x$ or $y+1=x$ and $\epsilon(x,y)$ is  $x \moo1 = y$,(3) $\epsilon(x,y)$ is $x=y$ and $o(x,y)$ is not $x=y$, (4) $\epsilon(x,y)$ is $x\farsim y$ and $o(x,y)$ is $x+1=y$ or $y+1=x$.  Here, in each of the cases, the premise of $\phi_i$ is false. Thus $\phi_i$ is always true and we take a set automaton that accepts all inputs.

Now we proceed to describe the construction for the remaining cases.

\noindent\textbf{Case 1: $x \ll y$ and $x\farsim y$.}

The formula $\phi_i$ states that if an $\alpha_i$ is followed by a class-distant $\beta$ that is not its successor, then the word bordered by them is in $L_i$.  
Observe that the condition $x\farsim y$ captures that $x+1\neq y$. Thus in the construction we do not need to specifically ensure that $x\ll y$ and can instead simply work with $x<y$. 
Here, in addition to the sets in $X$ and $A$, we take the sets $H_i=\{H_m\mid m\in M\}$, and $F_i=\{F_m\mid m\in M\}$, i.e., $Y_i=H_i\cup F_i$. 
The sets in $H_i$ (short for Hitherto) are used to compute the images of factors defined by a non-latest $\alpha_i$ and the latest $\alpha_i$ of a data value $d$. 
The sets in $F_i$ (short for henceForth) are used to compute the images of factors defined by the latest $\alpha_i$ of $d$ and the latest occurrence of $d$.
During the run of the automaton, the images of factors defined by an $\alpha_i$ and the current position is required to check whether the conclusion of the formula $\phi_i$ is ensured whenever the premise is met. 
This is be obtained from the sets by multiplying the images of factors stored in $H_i$, $F_i$, and $X$. 

Assume that the automaton is on a position $y$. 
Let $d$ be a data value and let $x_0$ be the latest occurrence of $d$ with respect to the current position if it exists, i.e., $x_0=\hat{d}(y)$.
Similarly, let $x_{0}^i$ be the latest $\alpha_i$ of $d$ with respect to the current position if it exists, i.e., $x_{0}^i=\hat{\alpha_i}^d(y)$. 
During the run of the automaton, the following invariants are maintained in addition to (X) and (A).
\begin{itemize}
\item[(H)] The data value $d$ is in $H_m$ if and only if there is a factor defined by a non-latest $\alpha_i$ and the latest $\alpha_i$ in $d$'s class whose image is $m$. That is to say, there is a factor $w_{(x,x_{0}^i)}$  for some $x \in \alpha_i^d(0,x_{0}^i)$ with the image $m$. 
\item[(F)] The data value $d$ is in $F_m$ if and only if there is a factor defined by the latest $\alpha_i$ in $d$'s class and the latest occurrence of $d$ whose image is $m$. That is to say, the factor $w_{(x_{0}^i,x_0)}$ has image $m$. 
\end{itemize}
On encountering a pair $(\ell,d')$ the automaton performs the following actions in addition to those described for the sets in $X$ and $A$ earlier.


\noindent\emph{If $\ell=\alpha_i$}:
The present data value $d'$ is removed from all the sets in $H$ and is added to the set $H_{u\cdot h(\alpha_i)\cdot v\cdot w}$ if prior to the set update, the data value was present in $H_u, F_{v}$ and $X_{w}$. 
Further, the data value is added to the set $H_m$ if prior to the set update, the data value was present in $X_m$ and $A_{\alpha_i}$.  
This is to maintain the invariant (H) in the upcoming positions. 
To maintain the invariant (F) in the upcoming positions, the data value is removed from the sets in $F$ if present. 

\noindent\emph{If $\ell\neq\alpha_i$}:
If the present data value $d'$ is present in any of the sets in $F$, it is removed from it and added to the set $F_{m\cdot h(\gamma)\cdot n}$ if prior to the set update, the data value was present in the sets $F_m, A_\gamma$, and $X_{n}$. 
Further, if the data value $d'$ was present in $X_m$ and $A_{\alpha_i}$ prior to the set update then it is added to the set $F_m$. 
This is to maintain the invariant (F).
\emph{If $\ell=\beta_i$}, the following operations are also performed. 
To satisfy $\phi_i$, we need to ensure that (1) the unique set $F_m$ containing $d'$ corresponds to an accepting element (i.e., $m \in P_i$) if $d'$ is not in $A_{\alpha_i}$ (i.e., if the latest element in the class is not an $\alpha_i$) (2) and all the sets $H_n$ containing $d'$ must satisfy $n\cdot h(\alpha_i)\cdot m \in P_i$. 
Note that this captures only the condition $x\ll y$ as it only considers $x\farsim y$ where $x<y$, thus ensuring $x+1\neq y$. 
This can be ascertained from the characteristic vector of $d'$. 
If it is not, then $\calA$ halts erroneously. 

Finally, at the end of the run, every configuration of the set automaton is accepting.

\smallskip
\noindent\textbf{Case 2: $x\ll y$ and $x\moo1 y$.}

The formula $\phi_i$ states the property that if an $\alpha_i$ is followed by a $\beta_i$ that is also the class successor of $\alpha_i$ and not the successor of $\alpha_i$, then the word bordered by them is in $L_i$. 
Here, the automaton does not require any sets $Y_i$ as the sets in $X$ already track the run of the latest data value and $A_{\alpha_i}$ tracks whether this occurred with an $\alpha_i$. 
Further, the automaton tracks with its state space whether the preceeding position was labelled by $\alpha_i$.

Let the automaton be at a position with data value $d'$. 
Assume that $d'$ exists in the set $X_m$ and $A_{\alpha_i}$ before the update. 
If $d'$ is labelled by a $\beta$ and $d'$ is not in the set $V$, then $m$ has to be in the accepting set $P_i$ to satisfy the formula. Otherwise, the automaton halts erroneously. 
However, if the preceeding position is labelled by $\alpha_i$ as given by the state space then the current position is the successor of $\alpha_i$ which does not satisfy the premise of the formula and thus the above condition is not required to be met. 
Finally, at the end of the run, every configuration of the set automaton is accepting.

\smallskip
\noindent\textbf{Case 3: $x+1= y$ and $x\moo1 y$.}
The formula $\phi_i$ states that for each $\alpha_i$ whose successor is a $\beta_i$ in the same class, the word bordered by them is in $L_i$. 
Here, the automaton does not require any sets $Y_i$ as the set $A_{\alpha_i}$ tracks whether the latest occurrence of the data value was with an $\alpha_i$.
Further, the automaton tracks with its state space whether the preceeding position was labelled by $\alpha_i$.

When reading a $\beta_i$ with data value $d'$ that is present in the set $A_{\alpha_i}$ and the preceeding position is labelled by $\alpha_i$ according to the state space, the set automaton halts erroneously if the empty word is not in $L_i$. This ensures that the formula is satisfied. 
Finally, at the end of the run, every configuration of the set automaton is accepting.

\smallskip
\noindent\textbf{Case 4 (5, 6): $y\ll x$ and $y\farsim x$ ($y\ll x$ and $y\moo1 = x$, $y+1=x$ and $y\moo1=x$).}
These cases are exactly the cases 1, 2, and 3 respectively but with the variables $x$ and $y$ swapped. Thus, the constructions in the cases 1, 2, and 3 follow with the change that everywhere in the construction, $\alpha_i$ and $\beta_i$ are replaced by each other. 

\end{proof}


\begin{lemmarep}\label{lem:AE-logic-to-aut}
    
    For each formula $\phi_i$ of the form
    \[(\forall\exists):\ \  \forall x \exists y\ (\alpha(x) \to \beta(y)\land o(x,y)\land \epsilon(x,y) \land \widetilde\lambda(x,y))\]
    the $h$-suffix-storing set automaton $\calA$ can be constructed such that, using the sets $X\cup A\cup Y_i$ it ensures that a data word is accepted by the automaton $\calA$ if and only if it satisfies the formula $\phi_i$. 
\end{lemmarep}
\begin{proofsketch}
We illustrate the key ideas of the construction by considering the case where $o(x,y)$ is $x \ll y$, $\epsilon(x,y)$ is $x\farsim y$, and $\widetilde\lambda(x,y)$ is $\widetilde{L}_i(x,y)$.

The formula $\phi_i$ states that if an $\alpha_i$ is followed by a class-distant $\beta$ that is not its successor, then the word bordered by them is in $L_i$.  
Observe that the condition $x\farsim y$ captures that $x+1\neq y$. Thus in the construction we do not need to specifically ensure that $x\ll y$ and can instead simply work with $x<y$. 
Here, in addition to the sets in $X$ and $A$, we take the sets $H_i=\{H_m\mid m\in M\}$, and $F_i=\{F_m\mid m\in M\}$, i.e., $Y_i=H_i\cup F_i$. 
The sets in $H_i$ (short for Hitherto) are used to compute the images of factors defined by a non-latest $\alpha_i$ and the latest $\alpha_i$ of a data value $d$. 
The sets in $F_i$ (short for henceForth) are used to compute the images of factors defined by the latest $\alpha_i$ of $d$ and the latest occurrence of $d$.
During the run of the automaton, the images of factors defined by an $\alpha_i$ and the current position is required to check whether the conclusion of the formula $\phi_i$ is ensured whenever the premise is met. 
This is be obtained from the sets by multiplying the images of factors stored in $H_i$, $F_i$, and $X$. 

Assume that the automaton is on a position $y$. 
Let $d$ be a data value and let $x_0$ be the latest occurrence of $d$ with respect to the current position if it exists, i.e., $x_0=\hat{d}(y)$.
Similarly, let $x_{0}^i$ be the latest $\alpha_i$ of $d$ with respect to the current position if it exists, i.e., $x_{0}^i=\hat{\alpha_i}^d(y)$. 
During the run of the automaton, the following invariants are maintained in addition to (X) and (A).
\begin{itemize}
\item[(H)] The data value $d$ is in $H_m$ if and only if there is a factor defined by a non-latest $\alpha_i$ and the latest $\alpha_i$ in $d$'s class whose image is $m$. That is to say, there is a factor $w_{(x,x_{0}^i)}$  for some $x \in \alpha_i^d(0,x_{0}^i)$ with the image $m$. 
\item[(F)] The data value $d$ is in $F_m$ if and only if there is a factor defined by the latest $\alpha_i$ in $d$'s class and the latest occurrence of $d$ whose image is $m$. That is to say, the factor $w_{(x_{0}^i,x_0)}$ has image $m$. 
\end{itemize}

\begin{example}
The invariants can be understood from the below figure. For ease of presentation, in \Cref{fig:word-logicaut}, for a pair of positions $i<j$ we use the notation $X_{(i,j)}$ to denote the set $X_m$ where $m$ is the image of the factor $w_{(i,j)}$ in the monoid $M$. 
\end{example}

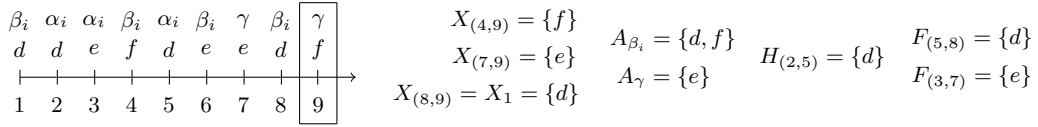
\begin{figure}[hbt!]
\centering
\resizebox{\columnwidth}{!}{
\input{word.tikz}
}
\caption{A data word $w$ and the contents of the sets of $\calA$ on encountering position $9$ of $w$.}
\label{fig:word-logicaut}
\end{figure}



On encountering a pair $(\ell,d')$ the automaton performs the following actions in addition to those described for the sets in $X$ and $A$ earlier.


\noindent\emph{If $\ell=\alpha_i$}:
The present data value $d'$ is removed from all the sets in $H$ and is added to the set $H_{u\cdot h(\alpha_i)\cdot v\cdot w}$ if prior to the set update, the data value was present in $H_u, F_{v}$ and $X_{w}$. 
Further, the data value is added to the set $H_m$ if prior to the set update, the data value was present in $X_m$ and $A_{\alpha_i}$.  
This is to maintain the invariant (H) in the upcoming positions. 
To maintain the invariant (F) in the upcoming positions, the data value is removed from the sets in $F$ if present. 

\noindent\emph{If $\ell\neq\alpha_i$}:
If the present data value $d'$ is present in any of the sets in $F$, it is removed from it and added to the set $F_{m\cdot h(\gamma)\cdot n}$ if prior to the set update, the data value was present in the sets $F_m, A_\gamma$, and $X_{n}$. 
Further, if the data value $d'$ was present in $X_m$ and $A_{\alpha_i}$ prior to the set update then it is added to the set $F_m$. 
This is to maintain the invariant (F).
\emph{If $\ell=\beta_i$}, the following operations are also performed. 
To satisfy $\phi_i$, we need to ensure that (1) the unique set $F_m$ containing $d'$ corresponds to an accepting element (i.e., $m \in P_i$) if $d'$ is not in $A_{\alpha_i}$ (i.e., if the latest element in the class is not an $\alpha_i$) (2) and all the sets $H_n$ containing $d'$ must satisfy $n\cdot h(\alpha_i)\cdot m \in P_i$. 
Note that this captures only the condition $x\ll y$ as it only considers $x\farsim y$ where $x<y$, thus ensuring $x+1\neq y$. 
This can be ascertained from the characteristic vector of $d'$. 
If it is not, then $\calA$ halts erroneously. 

Finally, at the end of the run, every configuration of the set automaton is accepting.
\end{proofsketch}
\begin{proof}
We first discuss the degenerate cases below. 

Consider the following cases (a) $\widetilde\lambda(x,y)$ is $\mathit{false}$, (b) $\epsilon(x,y)$ is $x\nsim y$, (c) $o(x,y)$ or $\epsilon(x,y)$ is $x=y$, (d) $o(x,y)$ is $x\ll y$ or $x+1=y$ and $\widetilde\lambda(x,y)$ is $\widetilde{L}_i(y,x)$, (e) $o(x,y)$ is $y\ll x$ or $y+1=x$ and $\widetilde\lambda(x,y)$ is $\widetilde{L}_i(x,y)$, (f) $o(x,y)$ is $x\ll y$ or $x+1=y$ and $\epsilon(x,y)$ is  $y \moo1 = x$, and (g) $o(x,y)$ is $y\ll x$ or $y+1=x$ and $\epsilon(x,y)$ is  $x \moo1 = y$, (h) $\epsilon(x,y)$ is $x\farsim y$ and $o(x,y)$ is $x+1=y$ or $y+1=x$.
Here, in each of these cases, the conclusion in the formula $\phi_i$ is false. 
Thus, no $\alpha_i$'s must occur in the data word for it to satisfy the formula. 
Therefore, we take a finite state automaton that ensures that no $\alpha_i$'s are present in the input. 
By the above cases (a), (d) and (e), henceforth whenever $o(x,y)$ is $x\ll y$ or $x+1=y$, we assume that $\widetilde\lambda(x,y)$ is $\widetilde{L}_i(x,y)$ and whenever $o(x,y)$ is $y\ll x$ or $y+1=x$, we assume $\widetilde\lambda(x,y)$ to be $\widetilde{L}_i(y,x)$.

Now we proceed to describe the construction for the remaining cases. 

\smallskip
\noindent\textbf{Case 1: $x \ll y$ and $x \farsim  y$.} 

The formula $\phi_i$ states for that each $\alpha_i$ there is a subsequent class-distant $\beta_i$ that is not its immediate successor and the word bordered by them is in $L_i$. 
We call such a $\beta_i$, a \emph{witness} of the $\alpha_i$ and at a position in the run, if such a $\beta_i$ has been observed, we say that the $\alpha_i$ has been \emph{witnessed}. 
Observe that the condition $x\farsim y$ captures that $x+1\neq y$. Thus in the construction we do not need to specifically ensure that $x\ll y$ and can instead simply work with $x<y$. 

Here, in addition to the sets in $X$ and $A$, we take the sets $H_i=\{H_m\mid m\in M\}$ and $F_i=\{F_m\mid m\in M\}$, i.e., $Y_i=H_i\cup F_i$.

Assume that the automaton is on a position $y$. 
Let $d$ be a data value and let $x_0$ be the latest occurrence of $d$ with respect to the current position if it exists, i.e., $x_0=\hat{d}(y)$.
Similarly, let $x_{0}^i$ be the latest $\alpha_i$ of $d$ with respect to the current position if it exists, i.e., $x_{0}^i=\hat{\alpha_i}^d(y)$. 
During the run of the automaton, the following invariants are maintained in addition to (X), and (A).
\begin{itemize}
\item[(H)] The data value $d$ is in $H_m$ if and only if there is a factor defined by a non-latest $\alpha_i$ that has not been witnessed and the latest $\alpha_i$ in $d$'s class whose image is $m$. That is to say, there is a factor $w_{(x,x_{0}^i)}$  for some $x \in \alpha_i^d(0,x_{0}^i)$ with the image $m$. 
\item[(F)] The data value $d$ is in $F_m$ if and only if there is a factor defined by the latest $\alpha_i$ of $d$ that has not been witnessed and the latest occurrence of $d$ whose image is $m$. That is to say, the factor $w_{(x_{0}^i,x_0)}$ has image $m$. 
\end{itemize}

On encountering a pair $(\ell,d')$ the automaton performs the following actions in addition to those described for the sets in $X$ and $A$ earlier.

\noindent\emph{If $\ell=\alpha_i$}:
The present data value $d'$ is removed from all the sets in $H_i$ and is added to the set $H_{u\cdot h(\alpha_i)\cdot v\cdot w}$ if prior to the set update, the data value was present in $H_u, F_{v}$ and $X_{w}$. 
Further, the data value is added to the set $H_m$ if prior to the set update, the data value was present in $X_m$ and $A_{\alpha_i}$.  
This is to maintain the invariant (H) in the upcoming positions. 
To maintain the invariant (F) in the upcoming positions, the data value is removed from the sets in $F$ if present. 

\noindent\emph{If $\ell\neq\alpha_i$}:
If the present data value $d'$ is present in any of the sets in $F$, it is removed from it and added to the set $F_{m\cdot h(\gamma)\cdot n}$ if prior to the set update, the data value was present in the sets $F_m, A_\gamma$, and $X_{n}$. 
Further, if the data value $d'$ was present in $X_m$ and $A_{\alpha_i}$ prior to the set update then it is added to the set $F_m$. 
This is to maintain the invariant (F) in the upcoming positions.

Additionally, \emph{If $\ell=\beta_i$}, the following operations are performed. 
If the latest element in the class is not an $\alpha_i$, i.e., if $d'$ is not in $A_{\alpha_i}$, then the $\beta_i$ serves as a witness to the latest $\alpha_i$, if present, corresponding to the set $F_m$ containing $d'$ such that $m\in P_i$. 
Thus, the data value $d'$ is removed from the set $F_m$ using local updates.
Further, the $\beta_i$ serves as a witness to the $\alpha_i$'s corresponding to the data value $d'$ present in a set $H_m$ with the data value also present in $F_n$ such that $m\cdot h(\alpha_i) \cdot n\in P_i$. 
Thus, the data value is removed from those sets in $H_i$ using local updates.
On reaching the end of the input, we check that all the $\alpha_i$'s are witnessed. 
This is done by the acceptance condition which requires that the sets $\{H_{m}, F_m \mid m  \in M\setminus P_i\}$ are empty. Further, the set $A_{\alpha_i}$ is required to be empty, signifying that there are no data values whose latest occurrence is with an $\alpha_i$ and is thus not witnessed.

\smallskip
\noindent\textbf{Case 2: $x \ll y$ and $x \moo 1 = y$.} 

The formula $\phi_i$ states that each $\alpha_i$ is witnessed by a $\beta_i$ in a later position that is not its successor and is also its class successor and the word bordered by them is in $L_i$. 
Here, the automaton does not require any sets $Y_i$ as the sets in $X$ already track the run of the latest data value and $A_{\alpha_i}$ tracks whether this occurred with an $\alpha_i$.
Further, the automaton tracks with its state space whether the preceeding position was labelled by $\alpha_i$.

When a data value appears that is present in $A_{\alpha_i}$ and one of the sets in $X$, it must be labelled with a $\beta_i$. Furthermore, prior to the update, it should be present in one of the sets in $\{X_m \mid m \in P_i\}$ and the preceeding position must not be labelled by $\alpha_i$ as given by the state space. This ensures that each $\alpha_i$ is witnessed by its class successor that is a $\beta_i$ and is not its immediate successor. 
Violation of any of these will result in the automaton halting erroneously.
Once the set automaton reaches the end of the input, we check that all the $\alpha_i$'s  are witnessed. 
This is done by the acceptance condition which requires that there is no data value present in the set $A_{\alpha_i}$. This is because if a data value $d$ is present in $A_{\alpha_i}$ then the latest occurrence of $d$ was with the label $\alpha_i$ and is thus not witnessed by a $\beta_i$.

\smallskip
\noindent\textbf{Case 3: $x+1= y$ and $x\moo1= y$.}
The formula $\phi_i$ states that each $\alpha_i$ is witnessed by its successor which is a $\beta_i$ in the same class and that the word bordered by them is in $L_i$. 
Here, the automaton does not require any sets $Y_i$ as the set $A_{\alpha_i}$ tracks whether the latest occurrence of the data value was with an $\alpha_i$.
Further, the automaton tracks with its state space whether the preceeding position was labelled by $\alpha_i$.

At each position immediately succeeding a position labelled by $\alpha_i$ as given by the state space, it must be the case that the data value is present in $A_{\alpha_i}$, the label is $\beta_i$, and the empty word belongs to the language $L_i$. This ensures that the $\alpha_i$ is witnessed. Violation of any of these will result in the automaton halting erroneously. Once the set automaton reaches the end of the input, we check that all the $\alpha_i$'s are witnessed. This is done by the acceptance condition which requires that there is no data value present in the set $A_{\alpha_i}$. This is because, if a data value $d$ is present in $A_{\alpha_i}$ with the run not yet halted, then the final occurrence was with the label $\alpha_i$ and is thus not witnessed by a $\beta_i$.

\smallskip
\noindent\textbf{Case 4: $y\ll x$ and $y\farsim x$.}
The formula $\phi_i$ states that each $\alpha_i$ is witnessed by an earlier class-distant $\beta_i$ that is not its immediate predecessor and the word bordered by them is in $L_i$.
Observe that the condition $y\farsim x$ captures that $y+1\neq x$. Thus in the construction we do not need to specifically ensure that $y\ll x$ and can instead simply work with $y<x$. 

Here, in addition to the sets in $X$ and $A$, we take the sets $H_i=\{H_m\mid m\in M\}$ and $F_i=\{F_m\mid m\in M\}$, i.e., $Y_i=H_i\cup F_i$.

Assume that the automaton is on a position $y$. 
Let $d$ be a data value and let $x_0$ be the latest occurrence of $d$ with respect to the current position if it exists, i.e., $x_0=\hat{d}(y)$.
Similarly, let $x_{0}^i$ be the latest $\beta_i$ of $d$ with respect to the current position if it exists, i.e., $x_{0}^i=\hat{\beta_i}^d(y)$. 
During the run of the automaton, the following invariants are maintained in addition to (X), and (A).
\begin{itemize}
\item[(H)] The data value $d$ is in $H_m$ if and only if there is a factor defined by a non-latest $\beta_i$ that has not been witnessed and the latest $\beta_i$ in $d$'s class whose image is $m$. That is to say, there is a factor $w_{(x,x_{0}^i)}$  for some $x \in \beta_i^d(0,x_{0}^i)$ with the image $m$. 
\item[(F)] The data value $d$ is in $F_m$ if and only if there is a factor defined by the latest $\beta_i$ of $d$ that has not been witnessed and the latest occurrence of $d$ whose image is $m$. That is to say, the factor $w_{(x_{0}^i,x_0)}$ has image $m$. 
\end{itemize}

On encountering a pair $(\ell,d')$ the automaton performs the following actions in addition to those described for the sets in $X$ and $A$ earlier.

\noindent\emph{If $\ell=\beta_i$}:
The present data value $d'$ is removed from all the sets in $H_i$ and is added to the set $H_{u\cdot h(\beta_i)\cdot v\cdot w}$ if prior to the set update, the data value was present in $H_u, F_{v}$ and $X_{w}$. 
Further, the data value is added to the set $H_m$ if prior to the set update, the data value was present in $X_m$ and $A_{\beta_i}$.  
This is to maintain the invariant (H) in the upcoming positions. 
To maintain the invariant (F) in the upcoming positions, the data value is removed from the sets in $F$ if present. 

\noindent\emph{If $\ell\neq\beta_i$}:
If the present data value $d'$ is present in any of the sets in $F$, it is removed from it and added to the set $F_{m\cdot h(\gamma)\cdot n}$ if prior to the set update, the data value was present in the sets $F_m, A_\gamma$, and $X_{n}$. 
Further, if the data value $d'$ was present in $X_m$ and $A_{\beta_i}$ prior to the set update then it is added to the set $F_m$. 
This is to maintain the invariant (F) in the upcoming positions.
Additionally, \emph{If $\ell=\alpha_i$}, the following operations are performed. 
The $\alpha_i$ must be witnessed by an earlier $\beta_i$ that is class-distant. Thus, the automaton checks if either the data value $d'$ is present in a set $F_m$ such that $m\in P_i$ or if it is present in sets $H_m$ and $F_n$ such that $m\cdot h(\beta_i) \cdot n\in P_i$. If either of the above are met, then the corresponding $\beta_i$ serves as a witness to the $\alpha_i$. Otherwise, the automaton halts erroneously.  
Finally, at the end of the input, every configuration of the automaton is accepting.

\smallskip
\noindent\textbf{Case 5: $y \ll x$ and $y \moo 1 = x$.} 

The formula $\phi_i$ states that each $\alpha_i$ is witnessed by its class predecessor which is a $\beta_i$ that is not its immediate predecessor and the word bordered by them is in $L_i$. 
Here, the automaton does not require any sets $Y_i$ as the sets in $X$ already track the run of the latest data value and $A_{\beta_i}$ tracks whether this occurred with an $\beta_i$.
Further, the automaton tracks with its state space whether the preceeding position was labelled by $\beta_i$.

At each position labelled by $\alpha_i$, the corresponding data value must be present in $A_{\beta_i}$. Furthermore, prior to the update, it should be present in one of the sets in $\{X_m \mid m \in P_i\}$ and the preceeding position must not be labelled by $\beta_i$ as given by the state space. This ensures that each $\alpha_i$ is witnessed by its class predecessor that is a $\beta_i$ and is not its immediate predecessor. 
Violation of any of these will result in the automaton halting erroneously.
Finally, at the end of the input, every configuration of the automaton is accepting.

\smallskip
\noindent\textbf{Case 6: $y+1= x$ and $y\moo1 =x$.}
The formula $\phi_i$ states that each $\alpha_i$ is witnessed by its predecessor which is a $\beta_i$ in the same class and that the word bordered by them is in $L_i$. 
Here, the automaton does not require any sets $Y_i$ as the set $A_{\beta_i}$ tracks whether the latest occurrence of the data value was with a $\beta_i$.
Further, the automaton tracks with its state space whether the preceeding position was labelled by $\beta_i$.

At each position labelled by $\alpha_i$, the corresponding data value must be present in $A_{\beta_i}$. Further, the preceeding position must be labelled by $\beta_i$ according to the state space. Finally, the empty word must be in the language $L_i$. Violation of any of these will result in the automaton halting erroneously. Finally, at the end of the input, every configuration of the automaton is accepting. 
\end{proof}

%% file: word.tikz
\begin{tikzpicture}
	\begin{pgfonlayer}{nodelayer}
		\node [style=none] (38) at (9, 0) {};
		\node [style=none] (65) at (1, 1.5) {\footnotesize $\beta_i$};
		\node [style=none] (66) at (1, 0.75) {\footnotesize $d$};
		\node [style=none] (67) at (2, 1.5) {\footnotesize $\alpha_i$};
		\node [style=none] (68) at (2, 0.75) {\footnotesize $d$};
		\node [style=none] (73) at (3, 1.5) {\footnotesize $\alpha_i$};
		\node [style=none] (74) at (3, 0.75) {\footnotesize $e$};
		\node [style=none] (75) at (4, 1.5) {\footnotesize $\beta_i$};
		\node [style=none] (76) at (4, 0.75) {\footnotesize $f$};
		\node [style=none] (77) at (5, 1.5) {\footnotesize $\alpha_i$};
		\node [style=none] (78) at (5, 0.75) {\footnotesize $d$};
		\node [style=none] (81) at (6, 1.5) {\footnotesize $\beta_i$};
		\node [style=none] (82) at (6, 0.75) {\footnotesize $e$};
		\node [style=none] (83) at (7, 1.5) {\footnotesize $\gamma$};
		\node [style=none] (84) at (7, 0.75) {\footnotesize $e$};
		\node [style=none] (85) at (8, 1.5) {\footnotesize $\beta_i$};
		\node [style=none] (86) at (8, 0.75) {\footnotesize $d$};
		\node [style=none] (89) at (9, 1.5) {\footnotesize $\gamma$};
		\node [style=none] (90) at (9, 0.75) {\footnotesize $f$};
		\node [style=none] (91) at (1, -0.75) {\footnotesize 1};
		\node [style=none] (92) at (2, -0.75) {\footnotesize 2};
		\node [style=none] (93) at (4, -0.75) {\footnotesize 4};
		\node [style=none] (95) at (3, -0.75) {\footnotesize 3};
		\node [style=none] (96) at (6, -0.75) {\footnotesize 6};
		\node [style=none] (97) at (5, -0.75) {\footnotesize 5};
		\node [style=none] (99) at (9, -0.75) {\footnotesize 9};
		\node [style=none] (100) at (7, -0.75) {\footnotesize 7};
		\node [style=none] (101) at (8, -0.75) {\footnotesize 8};
		\node [style=none] (102) at (8.5, 2) {};
		\node [style=none] (103) at (9.5, 2) {};
		\node [style=none] (104) at (9.5, -1.25) {};
		\node [style=none] (105) at (8.5, -1.25) {};
		\node [style=none] (106) at (10, 0) {};
		\node [style=none] (107) at (14.25, 1.5) {\footnotesize $X_{(4,9)} = \{f\}$};
		\node [style=none] (108) at (14.25, 0.5) {\footnotesize $X_{(7,9)} = \{e\}$};
		\node [style=none] (109) at (18.5, 1) {\footnotesize $A_{\beta_i} = \{d, f\}$};
		\node [style=none] (110) at (18.25, 0) {\footnotesize $A_{\gamma} = \{e\}$};
		\node [style=none] (113) at (26.5, 1) {\footnotesize $F_{(5,8)} = \{d\}$};
		\node [style=none] (114) at (26.5, 0) {\footnotesize $F_{(3,7)} = \{e\}$};
		\node [style=none] (115) at (22.5, 0.5) {\footnotesize $H_{(2,5)} = \{d\}$};
		\node [style=none] (116) at (13.5, -0.5) {\footnotesize $X_{(8,9)}=X_1=\{d\}$};
		\node [style=none] (119) at (1, 0) {};
		\node [style=none] (120) at (2, 0) {};
		\node [style=none] (121) at (3, 0) {};
		\node [style=none] (122) at (4, 0) {};
		\node [style=none] (123) at (5, 0) {};
		\node [style=none] (124) at (6, 0) {};
		\node [style=none] (125) at (7, 0) {};
		\node [style=none] (126) at (8, 0) {};
	\end{pgfonlayer}
	\begin{pgfonlayer}{edgelayer}
		\draw (105.center) to (102.center);
		\draw (102.center) to (103.center);
		\draw (103.center) to (104.center);
		\draw [in=360, out=180] (104.center) to (105.center);
		\draw [style=black pointer] (38.center) to (106.center);
		\draw [style=leftlength] (119.center) to (120.center);
		\draw [style=leftlength] (120.center) to (121.center);
		\draw [style=leftlength] (121.center) to (122.center);
		\draw [style=leftlength] (122.center) to (123.center);
		\draw [style=leftlength] (123.center) to (124.center);
		\draw [style=leftlength] (124.center) to (125.center);
		\draw [style=leftlength] (125.center) to (126.center);
		\draw [style=length] (126.center) to (38.center);
	\end{pgfonlayer}
\end{tikzpicture}

%% file: sec5-decidability.tex
\section{Decidability \& Undecidability Results on $\FO^2$ with Guarded Regular Predicates}
\nosectionappendix
\label{sec:decidability}

\begin{toappendix}
   \subsection*{Proofs for \Cref{sec:decidability}}
\end{toappendix}

\subsection{Linear Bands}
A \emph{band} is a semigroup in which every element is an idempotent. 

\begin{definition}[Linear Band]
A monoid $M$ is a \emph{linear band} if it is a band and the preorder relation $\leq_\calJ$ on $M$ is total. 
In other words, $M$
satisfies identities $x^2=x$ and 
$xyx=x \vee yxy = y$ for all $x,y\in M$.
\end{definition}

\begin{example}
Below are some examples and non-examples of linear bands. 

\begin{enumerate}
    \item The two-element monoid $U_1=\{1,0\}$ with a zero is a linear band.

\item Consider the band $U_1^2=\{1,r,s,0\}$ with the multiplication $rs=sr=0$ given in \Cref{fig:u1-u1sqr-n2}. 
This example shows that linear bands are not closed under direct products, and hence
they fail to form a pseudovariety.
\item Consider the null monoid $N_k=\{1,x_1, \ldots, x_k, 0\}$, with the multiplication $x_i \cdot x_j =0$, for all $i,j \in [k]$. 
The null monoid $N_2$ is in \Cref{fig:u1-u1sqr-n2}.
The monoid $N_k$ is not a linear band.
\end{enumerate}
\end{example}

\begin{figure}[hbt]
\centering
\resizebox{0.75\columnwidth}{!}{
\input{linbnd-u1-u1sqr-n2.tikz}
}
\caption{ (a) A linear band recognising the language $(a+b)^*b+a^*$ over the alphabet $\{a,b\}$, (b) A linear band $U_1$, (c) A  band $U_1^2$, and (d) a linear monoid  $N_2$. Here, (c) and (d) are not linear bands. Idempotents are denoted by a star in the above semigroups.}
\label{fig:u1-u1sqr-n2}
\end{figure}
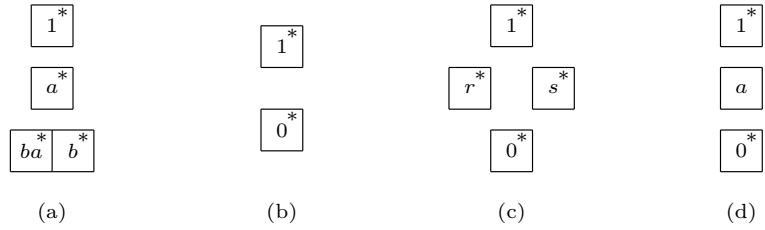

A semigroup $S$ is in the class $\bfD\bfA$ if it satisfies the identity $ese=e$ for each element $s$ and idempotent $e$ such that $e\leq_\calJ s$. The name $\bfD\bfA$ comes from an equiavalent definition that all regular $\calD$-classes of $S$ are aperiodic subsemigroups. 
They correspond to positive regular languages (subsets of $\Sigma^+$) recognised by $\fotwo(\Sigma,<)$ formulas as given in, for instance, the survey \cite{TT02}. 

\begin{lemma}\label{lem:linearband-da}
    Every linear band is in $\bfD\bfA$.
\end{lemma}
\begin{proof}
    Let $M$ be a linear band. Let $e,s\in M$ be such that $e\leq_\calJ s$. 
    Since $M$ is a linear band, we have $ese=e$ or $ses=s$. Let, if possible, $ese\neq e$. Then $ses=s$, that is, $s\leq_\calJ e$. By assumption we also have $e\leq_\calJ s$, therefore $e\calJ s$. 
    Since each element of $M$ is an idempotent, by location theorem (\Cref{Fact:location}) we have $es\in \calR(e)\cap \calL(s)$. 
    Further we have $ese\in \calR(es)\cap\calL(e)$. But $\calR(es)=\calR(e)$, therefore we have $ese\in\calH(e)$.
    Now by \Cref{Fact:uniqueIdemH} we have $ese=e$. 
\end{proof}

\subsection{Decidability with Guarded Regular Predicates recognised by Linear Bands}

In this section, we show that the logic $\EMSO^2[\Sigma,<,+1,\sim,\moo1, \widetilde\calF_M]$ has a decidable satisfiability problem when $M$ is a linear band. It suffices to show the below result. 


\begin{proposition}\label{prop:suffixstoring-oqnsa-linband}
Let $M$ be a linear band and $h:\Sigma^* \to M$ be a morphism. For each $h$-suffix-storing set automaton there is an equivalent ordered quasi-normal set automaton.
\end{proposition}
    
\begin{proof}
Let $M$ be a linear band and $\calA$ be a $h$-suffix-storing set automaton for some morphism $h:\Sigma^* \to M$. Assume that $\calX = \{X_m \mid m\in M\}$ are the non-stable sets of $\calA$.  
Since stable sets can always be added to an ordered quasi-normal set automaton, to simplify the below construction we ignore the stable sets. 

First, we introduce some machinery for the proof. 

Let $\sqsubseteq_{\calL}$ be a partial order on $M$ extending $\leq_{\calL}$, obtained by fixing total orders within each $\calJ$-class on $\calL$-classes such that they are stable on $\calR$-classes, and ordering distinct $\calJ$-classes according to $\leq_{\calL}$. More precisely, $\sqsubseteq_{\calL}$ satisfies the following conditions. Let $x,y\in M$. 
\begin{enumerate}
    \item If $x\sqsubset_\calL y$ then $x \leq_\calL y$.
    \item The relation $\sqsubseteq_\calL$ is antisymmetric, i.e., $x \sqsubseteq_\calL y$ and $y \sqsubseteq_\calL x$ then $x=y$.
    \item If $x <_\calL y$ then $x \sqsubset_\calL y$, and, if $x \calL y$ and $x\neq y$ then either $x \sqsubset_\calL y$ or $y \sqsubset_\calL x$.
    \item If  $x \sqsubseteq_\calL y$, $x \calR x'$ and $y \calR y'$, then $x' \sqsubseteq_\calL y'$.
\end{enumerate}
An order $\sqsubseteq_{\calL}$ is easily found by fixing an ordering on the columns of each $\calJ$-class in the eggbox diagram of $M$.

Let $x\sqsupseteq_\calL y$ and $x \sqsupset_\calL y$ denote $y \sqsubseteq_\calL x$ and $y \sqsubset_\calL x$ respectively.
A strict $\sqsubseteq_\calL$-chain in $M$ is a sequence $m_1,\ldots, m_n$ of elements from $M$ such that $m_1 \sqsupset_\calL m_2 \sqsupset_\calL \cdots \sqsupset_\calL m_n$. Let $k\geq 1$ be the length of the longest $\sqsubseteq_\calL$-chain in $M$.
The $\sqsubseteq_\calL$-height of an element $m\in M$, denoted by $k_m$, is the length of the longest $\sqsubseteq_\calL$-chain starting with $m$. 

We observe that a subset $M' \subseteq M$ is an $\leq_\calL$-total set, if, and only if, it is $\sqsubseteq_\calL$-total, and equivalently, the elements of $M'$ forms a strict $\sqsubseteq_\calL$-chain.
Moreover, if $m_1 \sqsupset_\calL m_2 \sqsupset_\calL \cdots \sqsupset_\calL m_n$ is a strict  $\sqsubseteq_\calL$-chain
then $k_{m_1} > k_{m_2} > \cdots > k_{m_n}$. 

\begin{claim}
\label{claim:M'm-Ltotal}
     If $M' \subseteq M$ is an $\leq_\calL$-total set. Then for all $m\in M$, $M'm$ is also $\leq_\calL$-total.
\end{claim}
\begin{claimproof}
    Assume that $M' \subseteq M$ is $\leq_\calL$-total. It suffices to prove that for each $m_1,m_2 \in M'$, if $m_1 \leq_\calL m_2$, then $m_1m \leq_\calL m_2m$, which follows from the fact that $\leq_\calL$ is stable under right multiplication (\Cref{fact:stable}).  
\end{claimproof}
\begin{claim}
\label{claim:R-height}
If $m \calR m'$  for $m,m' \in M$, then $k_m = k_{m'}$. 
\end{claim}
\begin{claimproof}
Assume that $m\calR m'$.
Let $m=m_1 \sqsupset_\calL m_2 \sqsupset_\calL \cdots \sqsupset_\calL m_n$ be a strict $\sqsubseteq_\calL$-chain starting in $m$. 
Let $z_i \in M$, for $1\leq i\leq n-1$, be such that $z_{i}m_{i}=m_{i+1}$.
The we claim that the $\sqsubseteq_\calL$-chain $m' \sqsupseteq_\calL z_1m'\sqsupseteq_\calL z_2z_1m'\cdots \sqsupseteq_\calL z_{n-1}\cdots z_1m'$ is strict. Since $m \calR m'$ we have $z_1m \calR z_1m'$. By induction,  we have $z_i \cdots z_1m \calR z_i \cdots z_1m'$, for $1\leq i \leq n-1$. Since $\sqsubseteq_\calL$ is compatible with the $\calR$-relation, we conclude that the chain is strict. Thus we infer that $k_m \leq k_{m'}$. The other direction follows symmetrically.
\end{claimproof}
\begin{claim}
\label{claim:ordered}
    Let $m_1,m_2 \in M$ be two elements such that $m_1 \sqsubset_\calL m_2$. Then for each $m\in M$, if $k_{m_1m} < k_{m_1}$ then $k_{m_2m} < k_{m_1}<k_{m_2}$.
\end{claim}
\begin{claimproof}
Assume that $k_{m_1m} < k_{m_1}$. We claim that $m_1m <_\calJ m_1$.
Otherwise $m_1m\calJ m_1$ and we also have $m_1m\leq_\calR m_1$, therefore by \Cref{Fact:fallToRLClass} we have $m_1m \calR m_1$. Now by \Cref{claim:R-height} we have $k_{m_1m} = k_{m_1}$, a contradiction. 
Hence we have $m_1m<_\calJ m_1$. Since $M$ is linear, we have either $m<_\calJ m_1$ or $m_1\leq_\calJ m$. If $m_1\leq_\calJ m$ we have $m_1mm_1=m_1$ since $M$ is in $\bfD\bfA$. Here we have $m_1\leq_\calJ m_1m$, a contradiction. 
Thus $m<_\calJ m_1$ and since $m_2m\leq_\calJ m$ we get $m_2m<_\calJ m_1$. 
We deduce that $k_{m_2m} < k_{m_1} < k_{m_2}$.
\end{claimproof}


We construct an ordered quasi-normal set automaton $\calB$ that has the family of sets $\calY \cup \calZ$, where $\calY=\{ Y_i \mid i \leq k\}$ and $\calZ =\{ Z_m \mid m \in M\}$. The sets in $\calY$ are used to simulate the non-stable sets $\calX$ of the $h$-suffix-storing set automaton $\calA$. 
At any point during a run of the automaton $\calB$, the set $Y_{k_m}$ has the content of a unique nonempty set $X_m$ if it exists.
During the run, except for the last step, the sets in $\calZ$ remain empty and the identity transformation is applied to them. At the last step we copy the nonempty sets in $\calY$ to the corresponding set in $\calZ$ (i.e., $Y_{k_m}$ to $Z_m$). The automaton $\calB$ simulates $\calA$ in the following way.

As already observed , at any point during the run, the subset $M'=\{ m\in M \mid X_m \neq \emptyset\}$ is an $\leq_\calL$-total set, and hence $\sqsubseteq_\calL$-total also. 
The contents of each $X_m$ is stored in the set $Y_{k_m}$ and the automaton $\calB$ also remembers in its state the information about the element $m$ that corresponds to the nonempty set $Y_{k_m}$ (this is a partial map $f:[k] \to I$ defined as $f:k_m \mapsto m$). To effect the global update $\rho=\{(X_m, X_{mm'})\mid m\in M\}$ of $\calA$, i.e., right multiplication by an element $m' \in M$, the automaton $\calB$ moves the data values in each nonempty set $Y_{k_m}$ to the set $Y_{k_{mm'}}$, i.e. by the global update $\rho'=\{(Y_{k_m}, Y_{k_{mm'}})\mid m\in M'\} \cup \{(Z_m,Z_m) \mid m\in M\}$, and replace the map $f$ by the map $f':k_{mm'}\mapsto mm'$. 
Note that since $M'm'$ is a $\sqsubseteq_\calL$-total set, elements of  $M'm'$ form a strict $\sqsubseteq_\calL$-chain, and therefore $k_{m_1m'} \neq k_{m_2m'}$ whenever $m_1m' \neq m_2m'$, $m_1,m_2\in M'$. Therefore each nonempty set $Y_i$ corresponds to precisely one nonempty set from $\calX$. On the last input pair the global update $\rho=\{(X_m, X_{mm'})\mid m\in M\}$ is simulated by the update 
$\rho'' =\{(Y_{k_m}, Z_{mm'})\mid m\in M\}$, i.e., we copy the sets in $\calY$ to sets in $\calZ$ using the map stored in the sets. Finally the automaton $\calB$ accepts if a final state is reached and the each data value is in an accepting configuration. 

Let $\lesssim$ be the order on the family  $\calY \cup \calZ$, given by $\calY \lnsim \calZ$, where $\calY$ is ordered by the relation $\{(Y_i, Y_j) \mid j \leq i\}$, and we fix an arbitrary order on $\calZ$. 
Using \Cref{claim:ordered}, it is straightforward to check that global updates of  $\calB$ are ordered with respect to $\lesssim$.
\end{proof}

Thus, by Propositions \ref{prop:logic-to-aut}, \ref{prop:suffixstoring-oqnsa-linband}, and \Cref{cor:oqnsa-dec}, we have the below result. 
\begin{theorem}
\label{thm:decidability}
    Satisfiability of $\emsotwo[\Sigma,<,+1,\sim, \moo1, \widetilde{\calF}_M]$ formulas over data words is decidable when $M$ is a linear band.
\end{theorem}

\subsection{Undecidability with Guarded Regular Predicates not recognised by Linear Bands}

\begin{proposition}
\label{prop:N-U2}
    If $M$ is not a linear band then either $U_1^2$ or $N_2$ divides $M$.  
\end{proposition}
\begin{proof}
Let $M$ be a monoid that is not a linear band. If $M$ contains an element $x$ that is not an idempotent then $M$ is not a band. Here $N_2$ divides the submonoid generated by $x$. Hence assume that all elements of $M$ are idempotent. Now since $M$ is not a linear band, there exist two idempotents $e$ and $f$ that are $\calJ$-incomparable. Here $U_1^2$ divides $M$. 
\end{proof}

\begin{propositionrep}\label{prop:undec}
Satisfiability of $\EMSO^2[\Sigma,<,+1, \sim, \moo1, \widetilde{\calF}]$ formulas over data words is undecidable for the following families of languages $\calF$: 
 \begin{enumerate}
     \item $\{\Sigma^*a\Sigma^*, \Sigma^*b\Sigma^*\}$, 
     where $a,b\in \Sigma$.
     \item  $\{(\Sigma\setminus \{a\})^*\,a\,(\Sigma\setminus \{a\})^*\}$,  
     where $a\in \Sigma$.
 \end{enumerate} 
\end{propositionrep}  
\begin{proof}
We use a reduction from the halting problem of two-counter machines that is known to be undecidable \cite{Minsky67}.

First, we recall two counter machines. Let $i\in\{1, 2\}$.

A two-counter machine~$\mach$ is a tuple $(Q, q_0, \Delta,q_f)$ where $Q$ is the set of states with an initial state $q_0$ and a final state $q_f$, $\Delta \subseteq Q \times \instr \times Q$ is the transition relation where $\instr =\{I_i, D_i, Z_i\}$.
The instruction $I_i$ increases the value of counter $i$ by one, $D_i$ decreases the value of counter $i$ by one, while $Z_i$ performs a zero test for counter $i$.
The final state $q_f$ is reached after performing the zero-tests $(q,Z_1,q_f)$, and $(q,Z_2,q_f)$.
While incrementing a counter can happen unconditionally, a counter can be decremented only when the counter value is non-zero.

A \emph{configuration} of a two-counter machine is of the form $(q,c_1,c_2)$ where $q \in Q$ is a state, $c_1$ and $c_2$ are the values of counter $1$ and counter $2$ respectively.
An \emph{accepting run} of a two-counter machine is a sequence of configurations from the initial configuration $(q_0,0,0)$ to the final configuration $(q_f,0,0)$.

\proofsubparagraph{Proof of (1):}

We encode the run of a given two-counter machine $\mach = (Q, q_0, \Delta,q_f)$ as a data word over the alphabet $\Delta$, the string projection of which corresponds to the sequence of transitions of the two-counter machine. We partition the alphabet $\Delta$ into six disjoint subsets, $\{\Delta_\alpha \subseteq \Delta \mid \alpha \in \{\iota_1, \delta_1, z_1, \iota_2, \delta_2, z_2\} \}$. Each transition incrementing counter $i$ is belongs to the class $\Delta_{\iota_i}$. Similarly, the transitions decrementing and zero-testing counter $i$ belong to $\Delta_{\delta_i}$ and $\Delta_{z_i}$ respectively.
The same data value is used when a counter's value is incremented from the value $c$ to $c+1$ and its subsequent decrement from $c+1$ to $c$. Apart from this, each transition is given a unique data value from the infinite domain. 

We write a formula $\varphi \in \EMSO^2[\Delta, <, +1, \sim, \moo1, \widetilde\calF]$ that is satisfied by the data word $w = (\gamma_1, d_1),\ldots,(\gamma_n, d_n)$, where $\gamma_i \in \Delta$ for $i\in[n]$, encoding the run of the two-counter machine $\mach$ if and only if the run is accepting.
The formula $\varphi$  is a conjunction of the following formulas describing sufficient properties for an accepting run of a two-counter machine $\mach$. 

\begin{enumerate}
	\item Each position is labelled by exactly one transition. We write 
\[\forall x \bigvee_{\gamma\in \Delta} \gamma(x)  \,\wedge\, \forall x   \bigwedge_{\substack{\gamma_i, \gamma_j \in \Delta,\\ \gamma_i \neq \gamma_j }} \neg (\gamma_i(x) \wedge \gamma_j(x)).\]
	
	\item The first position is labelled by a transition from an initial state. Let $\Delta_I \subseteq \Delta$ denote the set of transitions from an initial state. We write
\[\exists x \,\Biggl( (\forall y \, x <  y \vee x=y) \rightarrow \bigvee_{\gamma_j\in\Delta_I} \gamma_j(x) \Biggr).\]

	\item The last position is labelled by a transition to a final state. Let $\Delta_F \subseteq \Delta$ denote the set of transitions to a final state. We write
\[\exists x \,\Biggl( (\forall y \, y <  x \vee y=x) \rightarrow \bigvee_{\gamma_j\in\Delta_F} \gamma_j(x)\Biggr).\]

	\item We say two transitions $\gamma_i,\gamma_j\in \Delta$ are \emph{compatible}, denoted as $\mathit{comp}(\gamma_i,\gamma_j)$, if the target state of $\gamma_i$ is the source state of $\gamma_j$. Then, we state that consecutive positions are labelled by compatible transitions. We write
\[\forall x \forall y \, \Biggl( y = x+1 \wedge \gamma_i(x) \rightarrow \bigvee_{\substack{\gamma_i, \gamma_j\in\Delta,\\ \mathit{comp}(\gamma_i,\gamma_j)}} \gamma_j(y)\Biggr).\]

	\item Each class is a word in $\Delta_{\iota_i}\Delta_{\delta_i}$ or $\Delta_{z_i}$, for $i\in \{1,2\}$. We write 
\begin{align*}
\forall x \forall y \bigwedge\limits_{i\in\{1,2\}}& \, \Biggl( 
(x <y \wedge x \sim y  \rightarrow \bigvee_{\substack{\gamma_{\iota_i}\in \Delta_{\iota_i},\\ \gamma_{\delta_i}\in \Delta_{\delta_i}}} \gamma_{\iota_i}(x) \wedge  \gamma_{\delta_i}(y) ) \\
&\bigwedge \Biggl( (x \sim y  \rightarrow x = y) \rightarrow \bigvee_{\gamma_{z_i}\in \Delta_{z_i}} \gamma_{z_i}(x) \Biggr) \Biggr).
\end{align*}

	\item 
	There is no position with a label from $\Delta_{z_i}$ between two positions with labels from $\Delta_{\iota_i}$ and $\Delta_{\delta_i}$ of the same class, for $i\in\{1,2\}$. We write
\begin{align*}
\forall x \forall y \bigwedge\limits_{i\in\{1,2\}} \Biggl(  \bigvee_{\substack{\gamma_{\iota_i}\in \Delta_{\iota_i} \\\gamma_{\delta_i}\in\Delta_{\delta_i}}} \Biggl( \gamma_{\iota_i}(x) \wedge \gamma_{\delta_i}(y)  \vee \gamma_{\delta_i}(x) \wedge \gamma_{\iota_i}(y) \Biggr) \wedge x \sim y \rightarrow \neg \widetilde{\sem{\Delta^*\Delta_{z_i}\Delta^*}}(x,y)\Biggr).
\end{align*}

We observe that the languages $\Delta^*\Delta_{z_1}\Delta^*$ and $\Delta^*\Delta_{z_2}\Delta^*$ belong to the given family $\calF$ with $\Sigma=\Delta$, renaming $\Delta_{z_1}$ as $a$ and $\Delta_{z_2}$ as $b$.
\end{enumerate}

This concludes our proof.

\proofsubparagraph{Proof of (2):}

We encode the run of a given two-counter machine $\mach = (Q, q_0, \Delta,q_f)$ as a data word over the alphabet $\Sigma=\Delta\cup\{c_1,c_2\}$. The string projection of the data word is of the form $(\Delta c_1^*c_2^*)^*$ which corresponds to the sequence of transitions of the two-counter machine along with the number of $c_i$'s denoting the value of the counter $c_i$ after the transition. 

We partition the alphabet $\Delta$ into six disjoint subsets, $\{\Delta_\alpha \subseteq \Delta \mid \alpha \in \{\iota_1, \delta_1, z_1, \iota_2, \delta_2, z_2\} \}$. 
The same data value is used for a transition incrementing a counter, its corresponding $c_i$ present in each factor $\Delta c_1^*c_2^*$, and the transition eventually decrementing it. Thus, each class is of the form $\Delta_{\iota_i}c_i^+\Delta_{\delta_i}$ or $\Delta_{z_i}$ for the zero-test transitions, where $i\in\{1,2\}$. 

We write a formula $\varphi \in \EMSO^2[\Sigma, <, +1, \sim, \moo1,  \widetilde\calF]$ that is satisfied by the data word $w = (\gamma_1, d_1),\ldots,(\gamma_n, d_n)$, where $\gamma_i \in \Sigma$ for $i\in[n]$, encoding the run of the two-counter machine $\mach$ if and only if the run is accepting.
The formula $\varphi$ is a conjunction of the following formulas describing sufficient properties for an accepting run of a two-counter machine $\mach$.

\begin{enumerate}
	\item Each position is labelled by exactly one letter, i.e., 
\[\forall x \bigvee_{\gamma\in \Sigma} \gamma(x)  \,\wedge\, \forall x   \bigwedge_{\substack{\gamma_i, \gamma_j \in \Sigma,\\ \gamma_i \neq \gamma_j }} \neg (\gamma_i(x) \wedge \gamma_j(x)).\]

        \item The string projection is a word in $(\Delta c_1^*c_2^*)^*$. 
        We say two transitions $\gamma_i,\gamma_j\in \Delta$ are \emph{compatible}, denoted as $\mathit{comp}(\gamma_i,\gamma_j)$, if the target state of $\gamma_i$ is the source state of $\gamma_j$. Then, we state that two positions labelled by transitions in $\Delta$ with no such position in between, are labelled by compatible transitions. 
        This is written as an $\EMSO^2[\Sigma, <, +1]$ formula.

        \item Each class is a word in $\Delta_{\iota_i}c_i^+\Delta_{\delta_i}$ or $\Delta_{z_i}$ for $i\in\{1,2\}$. This is written as an $\EMSO^2[\Sigma, <, +1, \sim, \moo1]$ formula. 
    
	\item The first position is labelled by a transition from an initial state. Let $\Delta_I \subseteq \Delta$ denote the set of transitions from an initial state. Then, we write
\[\exists x \,\Biggl( (\forall y \, x <  y \vee x=y) \rightarrow \bigvee_{\gamma_j\in\Delta_I} \gamma_j(x) \Biggr).\]

	\item The last position is labelled by a transition to a final state. Let $\Delta_F \subseteq \Delta$ denote the set of transitions to a final state. Then, we write
\[\exists x \,\Biggl( (\forall y \, y <  x \vee y=x) \rightarrow \bigvee_{\gamma_j\in\Delta_F} \gamma_j(x)\Biggr).\]

        \item For each position with a label from $\Delta_{z_i}$, we ensure that the previous factor of the form $\Delta c_1^*c_2^*$  does not have any $c_i$'s. This is a regular property that can be written as an $\EMSO^2[\Sigma, <, +1]$ formula. 

	\item 
    Between two positions labelled by $c_i$ of the same class, there is exactly one position with label from $\Delta$, for $i\in\{1,2\}$.
    We write
\begin{align*}
\forall x \forall y \bigwedge\limits_{i\in\{1,2\}} &\Biggl( c_i(x) \land c_i(y) \land  x \moo1 =y \rightarrow \widetilde{\sem{(c_1+c_2)^*\Delta(c_1+c_2)^*}}(x,y)\Biggr).
\end{align*} 

We observe that the language $(c_1+c_2)^*\Delta(c_1+c_2)^*$ belongs to the given family $\calF$ with alphabet $\Sigma=\Delta\cup\{c_1,c_2\}$, renaming $\Delta$ as $a$.
\end{enumerate}

Hence, the result follows. 
\end{proof}

The below result follows from \Cref{prop:N-U2} and the fact that the language families in \Cref{prop:undec} can be recognised by the monoids $U_1^2$ and $N_2$ respectively.

\begin{theorem}
\label{thm:undecidability}
Satisfiability of $\emsotwo[\Sigma,<,+1,\sim, \moo1, \widetilde{\calF}_M]$ formulas over data words is undecidable when $M$ is not a linear band.
\end{theorem}

Our main result restated below follows from \Cref{thm:decidability}, \Cref{prop:N-U2} and 
\Cref{thm:undecidability}.%
\repeatedthm*%
%

%% file: linbnd-u1-u1sqr-n2.tikz
\begin{tikzpicture}
	\begin{pgfonlayer}{nodelayer}
		\node [style=none] (0) at (-1.5, -0.5) {};
		\node [style=none] (1) at (-0.5, -0.5) {};
		\node [style=none] (2) at (-1.5, 0.5) {};
		\node [style=none] (3) at (-0.5, 0.5) {};
		\node [style=none] (6) at (-0.5, 1) {};
		\node [style=none] (7) at (0.5, 1) {};
		\node [style=none] (8) at (0.5, 2) {};
		\node [style=none] (9) at (-0.5, 2) {};
		\node [style=none] (10) at (-1, 0) {};
		\node [style=none] (11) at (-1, 0) {\scriptsize$r$};
		\node [style=none] (13) at (0, 1.5) {};
		\node [style=none] (15) at (0, 1.5) {\scriptsize$1$};
		\node [style=none] (20) at (0.5, -0.5) {};
		\node [style=none] (21) at (1.5, -0.5) {};
		\node [style=none] (22) at (0.5, 0.5) {};
		\node [style=none] (23) at (1.5, 0.5) {};
		\node [style=none] (24) at (1, 0) {};
		\node [style=none] (25) at (1, 0) {\scriptsize$s$};
		\node [style=none] (30) at (-0.5, -2) {};
		\node [style=none] (31) at (0.5, -2) {};
		\node [style=none] (32) at (-0.5, -1) {};
		\node [style=none] (33) at (0.5, -1) {};
		\node [style=none] (34) at (0, -1.5) {};
		\node [style=none] (35) at (0, -1.5) {\scriptsize$0$};
		\node [style=none] (36) at (-6, 0.5) {};
		\node [style=none] (37) at (-5, 0.5) {};
		\node [style=none] (38) at (-5, 1.5) {};
		\node [style=none] (39) at (-6, 1.5) {};
		\node [style=none] (40) at (-5.5, 1) {};
		\node [style=none] (41) at (-5.5, 1) {\scriptsize$1$};
		\node [style=none] (42) at (-6, -1.5) {};
		\node [style=none] (43) at (-5, -1.5) {};
		\node [style=none] (44) at (-6, -0.5) {};
		\node [style=none] (45) at (-5, -0.5) {};
		\node [style=none] (46) at (-5.5, -1) {};
		\node [style=none] (47) at (-5.5, -1) {\scriptsize$0$};
		\node [style=none] (48) at (-5.5, -3) {};
		\node [style=none] (49) at (-5.5, -3) {};
		\node [style=none] (50) at (-5.5, -3) {\scriptsize (b)};
		\node [style=none] (51) at (0, -3) {};
		\node [style=none] (52) at (0, -3) {};
		\node [style=none] (53) at (0, -3) {\scriptsize (c) };
		\node [style=none] (54) at (5, 1) {};
		\node [style=none] (55) at (6, 1) {};
		\node [style=none] (56) at (6, 2) {};
		\node [style=none] (57) at (5, 2) {};
		\node [style=none] (58) at (5.5, 1.5) {};
		\node [style=none] (59) at (5.5, 1.5) {\scriptsize$1$};
		\node [style=none] (60) at (5, -0.5) {};
		\node [style=none] (61) at (6, -0.5) {};
		\node [style=none] (62) at (5, 0.5) {};
		\node [style=none] (63) at (6, 0.5) {};
		\node [style=none] (64) at (5.5, 0) {};
		\node [style=none] (65) at (5.5, 0) {\scriptsize$a$};
		\node [style=none] (66) at (5.5, -3) {};
		\node [style=none] (67) at (5.5, -3) {};
		\node [style=none] (68) at (5.5, -3) {\scriptsize (d)};
		\node [style=none] (69) at (5, -2) {};
		\node [style=none] (70) at (6, -2) {};
		\node [style=none] (71) at (5, -1) {};
		\node [style=none] (72) at (6, -1) {};
		\node [style=none] (73) at (5.5, -1.5) {};
		\node [style=none] (74) at (5.5, -1.5) {\scriptsize$0$};
		\node [style=none] (75) at (-11.5, 1) {};
		\node [style=none] (76) at (-10.5, 1) {};
		\node [style=none] (77) at (-11.5, 2) {};
		\node [style=none] (78) at (-10.5, 2) {};
		\node [style=none] (79) at (-11, 1.5) {};
		\node [style=none] (80) at (-11, 1.5) {\scriptsize$1$};
		\node [style=none] (81) at (-11.5, -0.5) {};
		\node [style=none] (82) at (-10.5, -0.5) {};
		\node [style=none] (83) at (-11.5, 0.5) {};
		\node [style=none] (84) at (-10.5, 0.5) {};
		\node [style=none] (85) at (-11, 0) {};
		\node [style=none] (86) at (-11, 0) {\scriptsize$a$};
		\node [style=none] (87) at (-12, -2) {};
		\node [style=none] (88) at (-11, -2) {};
		\node [style=none] (89) at (-12, -1) {};
		\node [style=none] (90) at (-11, -1) {};
		\node [style=none] (91) at (-11.5, -1.5) {};
		\node [style=none] (92) at (-11.5, -1.5) {\scriptsize$ba$};
		\node [style=none] (93) at (-11.25, -1.25) {\scriptsize *};
		\node [style=none] (94) at (-10, -1) {};
		\node [style=none] (95) at (-10, -2) {};
		\node [style=none] (96) at (-10.5, -1.5) {\scriptsize$b$};
		\node [style=none] (97) at (-10.75, 0.25) {\scriptsize *};
		\node [style=none] (98) at (-10.75, 1.75) {\scriptsize *};
		\node [style=none] (99) at (-10.25, -1.25) {\scriptsize *};
		\node [style=none] (101) at (-11, -3) {};
		\node [style=none] (102) at (-11, -3) {};
		\node [style=none] (103) at (-11, -3) {\scriptsize (a)};
		\node [style=none] (104) at (-5.25, 1.25) {};
		\node [style=none] (105) at (-5.25, 1.25) {\scriptsize *};
		\node [style=none] (106) at (-5.25, -0.75) {\scriptsize *};
		\node [style=none] (107) at (0.25, 1.75) {\scriptsize *};
		\node [style=none] (108) at (-0.75, 0.25) {\scriptsize *};
		\node [style=none] (109) at (1.25, 0.25) {\scriptsize *};
		\node [style=none] (110) at (0.25, -1.25) {\scriptsize *};
		\node [style=none] (111) at (5.75, 1.75) {\scriptsize *};
		\node [style=none] (113) at (5.75, -1.25) {\scriptsize *};
	\end{pgfonlayer}
	\begin{pgfonlayer}{edgelayer}
		\draw (2.center) to (0.center);
		\draw (2.center) to (3.center);
		\draw (3.center) to (1.center);
		\draw (1.center) to (0.center);
		\draw (6.center) to (9.center);
		\draw (9.center) to (8.center);
		\draw (8.center) to (7.center);
		\draw (7.center) to (6.center);
		\draw (22.center) to (20.center);
		\draw (22.center) to (23.center);
		\draw (23.center) to (21.center);
		\draw (21.center) to (20.center);
		\draw (32.center) to (30.center);
		\draw (32.center) to (33.center);
		\draw (33.center) to (31.center);
		\draw (31.center) to (30.center);
		\draw (36.center) to (39.center);
		\draw (39.center) to (38.center);
		\draw (38.center) to (37.center);
		\draw (37.center) to (36.center);
		\draw (44.center) to (42.center);
		\draw (44.center) to (45.center);
		\draw (45.center) to (43.center);
		\draw (43.center) to (42.center);
		\draw (54.center) to (57.center);
		\draw (57.center) to (56.center);
		\draw (56.center) to (55.center);
		\draw (55.center) to (54.center);
		\draw (62.center) to (60.center);
		\draw (62.center) to (63.center);
		\draw (63.center) to (61.center);
		\draw (61.center) to (60.center);
		\draw (71.center) to (69.center);
		\draw (71.center) to (72.center);
		\draw (72.center) to (70.center);
		\draw (70.center) to (69.center);
		\draw (77.center) to (75.center);
		\draw (77.center) to (78.center);
		\draw (78.center) to (76.center);
		\draw (76.center) to (75.center);
		\draw (83.center) to (81.center);
		\draw (83.center) to (84.center);
		\draw (84.center) to (82.center);
		\draw (82.center) to (81.center);
		\draw (89.center) to (87.center);
		\draw (89.center) to (90.center);
		\draw (90.center) to (88.center);
		\draw (88.center) to (87.center);
		\draw (88.center) to (95.center);
		\draw (94.center) to (95.center);
		\draw (94.center) to (90.center);
	\end{pgfonlayer}
\end{tikzpicture}

%% file: sec6-supplementary.tex
\section{Expressiveness} \label{sec:discussion}
\nosectionappendix

\begin{toappendix}
   \subsection*{Proofs for \Cref{sec:discussion}}
\end{toappendix}

The set automaton formalism is robust upto various acceptance conditions. 
\begin{lemmarep}
\label{lem:acceptance}
In a set automaton, the following acceptance conditions are equivalent:
Let $(q,\bfX)$ be the configuration reached at the end of the run.
\begin{enumerate}
\item $(F \subseteq Q, C \subseteq \bftwo^Y)$: $q \in F$
and the characteristic vector of each data value present in $\bfX$ is in $C$. 
\item $(F \subseteq Q, F_\ell \subseteq Q)$: $q \in F$
and at every class maximal position the state reached is in $F_\ell$.
\item $(F \subseteq Q, C \subseteq \bftwo^Y)$: $q \in F$
and at every class maximal position with data value $d$, the characteristic vector of $d$ is in $C$.
\item $(F \subseteq Q, S \subseteq Y)$: $q \in F$	
and the subset $S$ of the sets are empty.
\end{enumerate}
\end{lemmarep}
\begin{proof}
(1 $\Rightarrow$ 2): Let $\auto_1 = (Q,Y,A,\Delta,I,F,C)$, with acceptance condition (1).
The idea is to verify that at each class maximal position $i$, the characteristic vector $c_i$ of the corresponding data value is such that on applying the sequence of update relations $\rho_{i+1}\cdots \rho_n$ in the transitions of the rest of the run (since there are no more local updates affecting $d$), results in a characteristic vector present in the accepting set of configurations $C$.

We construct a finite state automaton $\calB$ to verify this and using this automaton we construct a set automaton with the acceptance condition (2).

Let $Q' = Q \times  \bfR_Y \times \bftwo^Y \times \bftwo$. This serves as the alphabet for the automaton $\calB$, tracking the current state, update relation, and the characteristic vector along the run of the set automaton $\calA_1$ respectively. 
The final element of the tuples in $Q'$ is used as a flag to track whether the characteristic vector, on performing the rest of the update relations in the run of $\calA_1$, belongs to $C$.
Consider the language 
\begin{align*}
    L = \{(q_1,\rho_1,&c_1, b_1) \cdots (q_n, \rho_n, c_n, b_n) \in Q'^* \mid b_n \! = \!1 \Leftrightarrow c_n \!\in\! C \\&\text{ and } b_i\! =\! 1 \Leftrightarrow (\rho_{i+1} \cdots \rho_{n})(c_i) \in C \text{ for $1\leq i\!<\! n$}\}
\end{align*}
The above language consists of all runs from a state $q_1$ to $q_n$ of the set automaton such that the characteristic vector at the end of the run is in $C$ as well as at every position $i$ along the run the characteristic vector is maintained such that the remaining set updates take it to the accepting set of characteristic vectors $C$.
The above language is regular and is accepted by a deterministic automaton that reads the input from right-to-left.
Let $\calB$ be the finite state automaton recognising this language over the alphabet $Q'$ with a finite set of states $Q_B$ and final states $F_B$.

We construct a set automaton $\auto_2 = (Q'\times Q_B,Y,A,\Delta',I,F', F_\ell)$ with the acceptance condition (2). We define the transitions $\Delta'$ in $\auto_2$ in the following way: For each transition $(q, a,\chi,\rho,\alpha,\beta,q') \in \Delta$, in $\auto_1$ we have the transitions 
$(p,a,\chi,\rho',\alpha,\beta,p')$ where $p = ((q,\rho,c,b),r) \in Q'\times Q_B$ ($\rho$, $c$, and $b$ are not relevant) and $p' = ((q', \rho', c', b'),s)$ is such that $b'\in \{0,1\}$, $c' \in \bftwo^Y$ is obtained by applying the set update $\rho'$ to $\chi$ and then modifying the resulting vector by the local update operations (i.e., $c'=\rho'(\chi) + \alpha -\beta$) and $s$ is a state reached after reading the tuple $(q', \rho', c', b')$ from the state $r$ in $\calB$. 
We let $F' = F \times \bfR_Y \times C \times \{1\}  \times F_B$, and $F_\ell = \{ (q,\rho,c,b) \mid b = 1\}\times F_B$.

(2 $\Rightarrow$ 3): We construct a new set automaton that has $|Q|$ (say $n$) additional sets $D_X=\{X_{q_1},\ldots, X_{q_n}\}$, along with the sets of the given automaton. The automaton preserves the property that at any point during the run a data value is present in at most one of the sets in $D_X$. After each transition with a target state $q$, the current data is removed from all the other sets in $D_X$ and is added to the set $X_q$. To finish the construction, we take $C$ to be the set of all vectors that has a $1$ in at least one of the components corresponding to the sets $\{ X_{q_i} \mid i \in [n], q_i \in F_\ell\}$. This ensures that on each class maximal position the state reached is in $F_\ell$.

\newcommand{\sinkt}{\mathtt{sink}}
(3 $\Rightarrow$ 4): For this translation, we construct a new set automaton that has two additional sets along with the sets of the given automaton.
Let's call the new sets a $\sinkt$ set and a set $X$. The sets $X$ and $\sinkt$ are always updated with the identity transformation.
When a data value appears for the first time, it is added to the set $X$.

For each class, the new automaton guesses the last occurrence of the data value, that is, the class maximal position of the data value.
At the class maximal position in the new automaton, the characteristic vector of the data present should be in $C$ (otherwise the automaton halts erroneously). 
Also at the class maximal position, in the corresponding transition of the new automaton, the data value is removed from every set in which the data value exists by local updates and is added to the $\sinkt$ set.
Thus, following the class maximal position, the data value should only be in the $\sinkt$ set.
If the new automaton guesses the class maximal position wrongly, this can be easily checked: When a data value is read that is in the $\sinkt$ set, it will lead to a non-accepting state.
In case a data value appears only once, it is directly added to the $\sinkt$ set rather than adding to the set $X$. At the end of the run the new automaton verifies that the set $X$ is empty, i.e., for each class a class maximal position has been guessed and verified.

(4 $\Rightarrow$ 1): To convert a set automaton with acceptance condition (4) to a set automaton with acceptance condition (1), we take $C \subseteq \bftwo^Y$ to be the set of all vectors whose $i^\text{th}$  components are 0 for each $i\in S$. 
\end{proof}

\subsection{Comparison with Other Models}

A natural restriction of set automata captures the class of languages recognised by Data automata.  

\begin{propositionrep}\label{prop:isa=cma}
For each data automaton there is an equivalent set automaton where each set is stable and thus having the update monoid be trivial. 
\end{propositionrep}
\begin{proof}
Data automata and class memory automata are known to be expressively equivalent \cite{BS07} with an effective translation in both directions. 
Therefore it suffices to consider class memory automata in place of data automata. 

First we recall class memory automata introduced in \cite{BS07}. 
Formally, a class memory automaton (CMA) $\calA = (Q, \Sigma, \Delta, q_0, F_\ell, F_g)$ where $Q$ is the finite set of states, $q_0$ is the initial state, $F_\ell \subseteq F_g \subseteq Q$ are the local and global accepting states respectively. 

The CMA has a hash function $h : D \to Q\cup\{\bot\}$ such that $h(d)\neq\bot$ for only finitely many $d$, which are precisely those that appear in the run of an input data word.
The hash function $h$ holds the state the automaton moved to when the data value $d$ was read the last time. 
Each transition $\delta$ of the form $(p, \ell, s, q)$ on the input $(\ell, d)$, if the current hash value of $d$ is $s$, takes the state from $p$ to $q$ and updates the hash value of $d$ to $q$. 
Initially, the hash value of every data value $d$ is set to $\bot$.
In an accepting run, the hash value for every data value $d$ at the end of the run is in $F_\ell \cup \{\bot\}$ and the state is in $F_g$.

First, we show that for every CMA, there is an equivalent identity set automaton. 
Consider a CMA $\calA = (Q, \Sigma, \Delta, q_0, F_\ell, F_g)$. 
Let $|Q| = k$ be the number of states of the CMA $\calA$.
Let $\calB$ be the required identity set automaton consisting of $\log_2 (k)$ number of sets with the same state space $Q$ as the CMA $\calA$. The set automaton $\calB$ when simulating the CMA $\calA$ remembers the local state of $\calA$ using the $\log_2(k)$ sets.
Each local state of the CMA is denoted using a binary encoding of length $\log_2(k)$. For instance, if $k=4$, the local states $q_0, q_1, q_2, q_3$ shall be represented as $00, 01, 10$, and $11$, respectively. 
Let $d$ be a data value present in the data word $w$. 
The binary encoding of the local state of $d$ in $\calA$ is given by its characteristic vector of membership in the set automaton $\calB$.

A transition of the form $(p, a, s, q)$ on the input $(a, d)$ is simulated by the set automaton $\calB$ as follows. 
Checking whether the current data value $d$ is in the local state $s$ is done by enabling the transition only if the characteristic vector $\chi$ for  membership check denotes the binary encoding of the state $s$. 
The local state of $d$ is updated to $q$ by updating the sets so that the new characteristic vector $\chi'$ for membership of $d$ denotes the binary encoding of the state $q$. 
Note that these do not require any global updates as they can be accomplished using the local updates $\alpha$ and $\beta$ to add and delete the current data value at the necessary sets. Thus, each set in the automaton $\calB$ is stable. 
This set automaton, using the acceptance condition (2) given in \Cref{lem:acceptance}, simulates the local acceptance condition of the CMA $\calA$ through the set of accepting characteristic vectors $C$. The accepting set of states is $F=F_g\subseteq Q$ in the set automaton $\calB$.

Thus, we can simulate accepting runs of the class memory automaton $\calA$ using an equivalent set automaton $\calB$ where all the sets are stable.
%
\end{proof}



We have the following result. 
\begin{propositionrep}\label{prop:sa=ca}
Set automata are expressively equivalent to class automata. 
\end{propositionrep}
\begin{proof}

We recall the definition of a class automaton \cite{BL10}. 

A \emph{class automaton} over the alphabet $\Sigma $ is a tuple $\calA=(\calB, \calC)$ where $\calB$ is a nondeterministic letter-to-letter transducer with input alphabet $\Sigma$ and output alphabet $\Gamma$, and $\calC$ is a nondeterministic finite automaton over $\Gamma\times \{0,1\}$.
A run of a class automaton on a data word $w=(a_1,d_1)\cdots (a_n, d_n)$ is \emph{accepting} if there is an accepting run of $\calB$ on the string projection $\mathit{str}(w)$, yielding the output $w'= b_1\cdots b_n\in \Gamma^*$, such that for each class $X$ of $w$, the class string $w'_X\in (\Gamma\times\{0,1\})^*$  obtained by labelling $w'$ by $1$ at positions in $X$ and $0$ everywhere else, is accepted by the automaton $\calC$. 
Here, the word $w'_X$ is called the \emph{class string} of $X$.

First, we show that for every set automaton, there is an equivalent class automaton.
Consider a set automaton $\calS = (Q, Y, A, \Delta, I, F, F_l)$ with the acceptance condition (2) from \Cref{lem:acceptance} recognising the language $L$.
We sketch the construction of a class automaton $\calA=(\calB, \calC)$ recognising $L$. 
The set automaton $\calS$ has the following regular properties for a run:
(1) the first position is in an initial state, (2) at the final position, there is a transition to an accepting state, and (3) for each pair of consecutive positions $x$ and $y$, there is a transition from the state at position $x$ to the state at position $y$. 

Such regular properties in the run of a set automaton are verified by the transducer $\calB$.
Recall that $\calB$ is a nondeterministic letter-to-letter transducer with input alphabet $\Sigma$ and output alphabet $\Gamma$.
The transducer $\calB$ does not have access to the data values in the input word of the set automaton $\calS$ and reads only the string projection.
Thus $\calB$ nondeterministically guesses the characteristic vector of membership of the current data value at each position and simulates the run of $\calS$. 
More precisely, on simulating a transition $(p, \ell, \chi, \rho, \alpha, \beta, q)$ of $\calS$, the transducer outputs $(\chi, \rho, \alpha, \beta)$. 
Thus, the output alphabet of $\calB$ is $\Gamma = \bftwo^Y \times \bfR_Y \times \bftwo^Y \times \bftwo^Y$.
The state space of $\calB$ is the same as that of $\calS$ and hence is $Q$.

The idea is that for the transducer $\calB$ to ensure that the regular constraints are satisfied while the NFA $\calC$ ensures that the class constraints are met.

The NFA $\calC$, on each class string, verifies that (1) the class-minimal position has a transition on the zero vector, (2) at the class-maximal position, there is a transition to a local final state, 
and (3) the transitions on successive positions in a class are on characteristic vectors which are consistent as per the update relations in the interval between the two positions.
That is, for any two positions $x$ and $y$ such that $y=x\moo1$, the vector $\chi_y$ at position $y$ must satisfy $\chi_y = M(\rho_{x+1} \rho_{x+2}\cdots \rho_{y-1})^T(\chi'_x)$ where $\chi'_x=M(\rho_x)^T \chi+\alpha_x-\beta_x$. 
Here, $(\chi_x, \rho_x, \alpha_x, \beta_x)$ is the label on the position $x$. 
The above properties can be recognised by a NFA and hence we can see that the class automaton $\calC$ accepts $L$. 

Now, we show that for every class automaton, there is an equivalent set automaton.
Consider a class automaton $\calA=(\calB, \calC)$ recognising $L$, with $\Gamma$ as the output alphabet of the transducer $\calB$. 
The states of the constructed set automaton $\calS$ keep track of how the current state of the transducer $\calB$ evolves over a run.
Thus when the transducer reaches an accepting state after reading the string projection of the input word, the same is reflected in the run of the set automaton $\calS$.
Since the number of classes in an input data word can be unbounded, the runs of $\calC$ on the class strings cannot be simulated in the finite state space of $\calS$.
The sets of $\calS$ are used to track this information.
Let $k$ be the number of states of the NFA $\calC$.
The set automaton $\calS$ uses $k$ sets $X_1, \dots, X_k$ each corresponding to a state of $\calC$.

Suppose on a position with $(\ell, d)$ in the input word, the transducer reads $\ell\in \Sigma$ and outputs $a\in \Gamma$. 
In the class strings that the NFA $\calC$ reads, the corresponding position has can thus have either $(a,0)$ or $(a,1)$. 
Thus there are bounded number of runs of the NFA $\calC$ that $\calS$ simulates when reading a letter-data pair.
The set automaton, at each position, makes a global update followed by a local update to simulate the run of the class automaton corresponding to all possible class strings.
At any position in the run of the set automaton $\calS$, a data word $d$ being present in a set $X$ means that the run of the NFA $\calC$ on the class string of $d$ until that position is at the state corresponding to the set $X$. 
Consider the case where the NFA reads $(a,0)$ and makes a transition from state $p$ to $p'$.
This is simulated in $\calS$ by moving the content of set $X_p$ to set $X_{p'}$ by a global update.
On the other hand, $\calC$ reading $(a,1)$ and making a transition from $p$ to $p''$ is simulated by a local update where only the current data value is moved from the set $X_{p'}$ to set $X_{p''}$.
This is because the transition on $(a,1)$ occurs only in the class string corresponding to the current data value and hence it can be simulated using local updates for the data value. 

Further, the set automaton $\calS$ also needs to remember in its state space the state $q$ that the NFA $\calC$ reaches on reading the class-minimal position on the class string so that the first time a data value $d$ is encountered, it is added to the set $X_q$ using local updates.
Note that starting from the initial state of $\calC$, the class string that takes $\calC$ to state $q$ has all $0$'s when projected to the data values. 
Let $Q_B, Q_C$, and $Q_S$ be the state space of the transducer $\calB$, NFA $\calC$, and set automata $\calS$ respectively.
We have that $Q_S = Q_B \times Q_C$.

The set automaton $\calS$ given by the ideas sketched above is the required set automaton recognising $L$. 
\end{proof}




\subsection{Logical Characterisation of Set Automata}

First, we recall a logical characterisation of class automata given in \cite{BL10}. 

\begin{proposition}[Boja\'nczyk and Lasota, \cite{BL10}]\label{prop:ca=rmso}
    Class automata are expressively equivalent to a restricted fragment of $\mathrm{MSO}[\Sigma, <, \sim]$ consisting of formulas of the form 
    \begin{equation}\label{eq:rmso}
        \exists X_1\cdots\exists X_n\,\forall\, Y\, \mathsf{class}(Y) \rightarrow \phi(X_1,\ldots,X_n,Y)
    \end{equation}
    where $\mathsf{class}(Y):=\exists y\forall x\, (Y(x)\leftrightarrow y\sim x)$ and $\phi$ is in $\mathrm{MSO}[\Sigma, <]$.
\end{proposition}

It is easy to see that the data value equality test, the class successor relation, and regular predicates are captured in the restricted fragment of MSO. 
\begin{proposition}\label{prop:emso-rmso}
    For each formula $\varphi\in\emsotwo[\Sigma, <, +1, \sim, \moo1, {\calF}]$, there is a formula $\psi$ belonging to the restricted fragment of $\mathrm{MSO}[\Sigma, <,\sim]$ consisting of formulas of the form given in \Cref{eq:rmso} such that $L(\varphi)=L(\psi)$.
\end{proposition}

By Propositions \ref{prop:aut-to-logic}, \ref{prop:sa=ca}, \ref{prop:ca=rmso}, and \ref{prop:emso-rmso}, we obtain the below result. 
\begin{theorem}
    Set automata and $\emsotwo[\Sigma, <, +1, \sim, \moo1, {\calF}]$ are expressively equivalent. 
\end{theorem}
Now by \Cref{prop:aut-to-logic}, we obtain the below corollary. 
\begin{corollary}\label{cor:unguarded-guarded}
    For each formula in $\EMSO^2[\Sigma,<,+1,\sim,\moo1,\calF_M]$, there is an equivalent formula in $\EMSO^2[\Sigma,<,+1,\sim,\moo1,\widetilde{\calF}_{M'}]$.
\end{corollary}
In fact, it can be shown that $M'=\calP(M)$ but we skip it here for the sake of brevity.

%% file: sec7-conclusion.tex
\section{Conclusion}
We have characterised the decidability frontier for the $\emsotwo$ logic over data words with guarded regular predicates recognised by a monoid. The characterisation in the case of guarded regular predicates recognised by semigroups is yet to be done. The class of regular languages recognised by linear bands has not been studied in the literature and needs to be explored further. 
\label{sec:conclusion}